\pdfoutput=1
\documentclass[10pt,pra,aps,superscriptaddress,twocolumn]{revtex4-1}
\usepackage{amsmath,bbm}
\usepackage{amsthm}
\usepackage{amsmath}
\usepackage{latexsym}
\usepackage{amssymb}
\usepackage{float}
\usepackage{graphicx}   
\usepackage{color}
\usepackage{mathpazo}
\usepackage{comment}
\usepackage{enumitem}
\usepackage{multirow}
\usepackage[colorlinks=true,linkcolor=blue,citecolor=magenta,urlcolor=blue]{hyperref}

\newcommand{\be}{\begin{equation}}
\newcommand{\ee}{\end{equation}}
\newcommand{\bea}{\begin{eqnarray}}
\newcommand{\eea}{\end{eqnarray}}

\def\squareforqed{\hbox{\rlap{$\sqcap$}$\sqcup$}}
\def\qed{\ifmmode\squareforqed\else{\unskip\nobreak\hfil
\penalty50\hskip1em\null\nobreak\hfil\squareforqed
\parfillskip=0pt\finalhyphendemerits=0\endgraf}\fi}
\def\endenv{\ifmmode\;\else{\unskip\nobreak\hfil
\penalty50\hskip1em\null\nobreak\hfil\;
\parfillskip=0pt\finalhyphendemerits=0\endgraf}\fi}

\newcommand{\tr}{\text{tr}}
\newcommand{\I}{\mathbbm{1}}

\newcommand{\la}{\langle}
\newcommand{\ra}{\rangle}

\newcommand{\cell}[1]{\begin{tabular}{@{}l@{}}#1\end{tabular}}

\makeatletter
\newtheorem*{rep@theorem}{\rep@title}
\newcommand{\newreptheorem}[2]{%
\newenvironment{rep#1}[1]{%
 \def\rep@title{#2 \ref{##1}}%
 \begin{rep@theorem}}%
 {\end{rep@theorem}}}
\makeatother

\newtheorem{thm}{Theorem}
\newreptheorem{thm}{Theorem}
\newtheorem{lemma}{Lemma}
\newtheorem{definition}{Definition}

\newtheorem{result}{Result}

\begin{document}



\title{Quantum contextuality provides communication complexity advantage}


\author{Shashank Gupta}
\affiliation{S. N. Bose National Centre for Basic Sciences,
Block JD, Sector III, Salt Lake, Kolkata 700106, India}
\author{Debashis Saha}
\affiliation{S. N. Bose National Centre for Basic Sciences,
Block JD, Sector III, Salt Lake, Kolkata 700106, India}
\affiliation{School of Physics, Indian Institute of Science Education and Research Thiruvananthapuram, Kerala 695551, India}
\author{Zhen-Peng Xu}
\affiliation{School of Physics and Optoelectronics Engineering, Anhui University, 230601 Hefei, People’s Republic of China}
\affiliation{Naturwissenschaftlich-Technische Fakult\"{a}t, Universit\"{a}t Siegen,
Walter-Flex-Stra{\ss}e 3, 57068 Siegen, Germany}
\author{Ad\'an~Cabello}
\affiliation{Departamento de F\'{\i}sica Aplicada II, Universidad de Sevilla, E-41012 Sevilla,
Spain}
\affiliation{Instituto Carlos~I de F\'{\i}sica Te\'orica y Computacional, Universidad de
Sevilla, E-41012 Sevilla, Spain}
\author{A. S. Majumdar}
\affiliation{S. N. Bose National Centre for Basic Sciences,
Block JD, Sector III, Salt Lake, Kolkata 700106, India}


\begin{abstract}
Despite the conceptual importance of contextuality in quantum mechanics, there is a hitherto limited number of applications requiring contextuality but not entanglement.
Here, we show that for any quantum state and observables of sufficiently small dimensions producing contextuality, there exists a communication task with quantum advantage. Conversely, any quantum advantage in this task admits a proof of contextuality whenever an additional condition holds. We further show that given any set of observables allowing for quantum state-independent contextuality, there exists a class of communication tasks wherein the difference between classical and quantum communication complexities increases as the number of inputs grows. Finally, we show how to convert each of these communication tasks into a semi-device-independent protocol for quantum key distribution.
\end{abstract}


\maketitle


\textit{Introduction.---}Contextuality is one of the most significant properties of quantum mechanics \cite{Bell:1966RMP,Kochen:1967JMM,Klyachko:2008PRL,Cabello2008,Cabello2014,Budroni:RMP2022}. It stipulates that, for some correlations, there is no probability distribution in agreement with the marginal distributions corresponding to sets of compatible (i.e., jointly measurable) observables. In particular, contextuality forbids us to assign predetermined context-independent values to the outcomes of quantum sharp measurements (defined as those that yield the same outcome when they are repeated and do not disturb any compatible observable). 
While nonlocality which can be seen as a form of quantum contextuality requiring entanglement, has found many applications in quantum communication \cite{Ekert1991,RevModPhys.82.665,Cubitt2010}, so far, entanglement unassisted quantum contextuality has found few applications despite of its conceptual importance
\cite{Kleinmann2011,Howard2014,Cabello2018,Grudka2014,Singh2017,Saha2019,bharti,Saha2020,PRXQuantum2022} \footnote{There is other notion of contextuality, preparation contextuality \cite{Spekkens2005}, which has found application in oblivious multiplexing and state discrimination tasks \cite{Spekkens2009,Pan2018,PRXSchmid}}.

Here, we first show that any contextual correlations achieved using quantum systems of sufficiently small dimensions offer a quantum advantage in a suitably designed one-way communication complexity (or distributed computation) task. Conversely, whenever an additional condition holds, any quantum protocol providing advantage in those tasks produces a proof of contextuality. By itself, this result provides an operational way to understand the sense in which some famous forms of quantum contextuality (notably, the one produced by the violation of the Klyachko-Can-Binicio\u{g}lu-Shumovsky inequality with quantum systems of dimension three \cite{Klyachko:2008PRL}) are ``nonclassical''. 

As a second result, we show that for every form of state-independent (SI) contextuality \cite{Cabello2008,PhysRevLett.103.050401,Yu2012,PhysRevLett.109.250402}, the ratio between the dimensions of the classical systems and quantum systems required to accomplish the task can be made arbitrarily large by increasing the number of inputs. These communication complexity tasks are the so-called \textit{equality problems} that appear in many practical scenarios \cite{ccbook,TRbook,rao_yehudayoff_2020}. 
Finally, we present a semi-device-independent (SDI) protocol for quantum key distribution (QKD) \cite{Marcin2011} based on the quantum advantage in our communication complexity tasks, in which security
is proven by using the monogamy relation \cite{Ramanathan2012,Kurzy2014,Saha2017} of contextuality.

\textit{Contextuality witnesses.---}Given a set $\{e_i\}_{i=1}^n$ of events produced in a contextuality experiment, one can define an $n$-vertex graph $G$ in which each event is represented by a vertex and exclusive events correspond to adjacent vertices. $G$ is called the graph of exclusivity of $\{e_i\}_{i=1}^n$. 
In quantum mechanics, each event $e_i$ is represented by a projector $\Pi_i$. Mutually exclusive events are represented by mutually orthogonal projectors. A quantum realization of a set of events $\{e_i\}_{i=1}^n$ with graph of exclusivity $G$ is a set of projectors $\{\Pi_i\}_{i=1}^n$ that satisfies all the exclusivity relations in $G$ and all the constraints imposed by the definition of the events.

\begin{definition}[Contextuality witness]
A functional
\be
\label{wit}
W= \sum_{i=1}^n w_i P(e_i), 
\ee
where $w_i \geqslant 0$ and $P(e_i)$ is the probability of event $e_i$, is a quantum contextuality witness if there is a quantum realization $\{\Pi_i\}_{i=1}^n$ of $\{e_i\}_{i=0}^n$ and a quantum state $\rho$ such that 
\be \label{wc}
\sum_{i = 1}^n w_i\ \tr (\rho \Pi_i) > \alpha (G,\Vec{w}),
\ee where 
$\alpha(G,\Vec{w})$
is the independence number of the vertex-weighted graph $(G,\Vec{w})$, where $G$ is the graph of exclusivity of $\{e_i\}_{i=0}^n$ and $\Vec{w} = \{w_i\}_{i=1}^n$. That is, $\alpha(G,\Vec{w})$ is the largest value of $\sum_{i \in I} w_i$, where $I$ is the set of the subsets consisting of nonadjacent vertices of $G$ \cite{bondybook}. 
\end{definition}

The name `contextuality witness' follows from the fact that, given $W$, one can find a noncontextuality inequality \cite{PhysRevA.93.032102,PhysRevLett.127.070401} whose upper bound for noncontextual models is $\alpha (G,\Vec{w})$ and whose quantum value is the left-hand side of Eq. \eqref{wc} \cite{PhysRevA.93.032102,PhysRevLett.127.070401,Cabello2014}. 

We will focus on quantum realizations of contextuality witnesses constructed as follows. We first identify a vertex-weighted graph $(G,\Vec{w})$ 
for which we can identify $\{\Pi_i\}_{i=1}^n$ and $\rho$ such that Eq. \eqref{wc} holds. We will refer to $\{(G,\Vec{w}),\{\Pi_i\}_{i=1}^n,\rho\}$ as a quantum realization of a contextuality witness for $(G,\Vec{w})$. 
In some cases, there is no need to identify a state $\rho$.

%
\begin{definition}[State-independent contextuality witness]
The functional \eqref{wit} is a quantum state-independent contextuality witness for dimension $d$ if there is a quantum realization $\{\Pi_i\}_{i=1}^n$ of $\{e_i\}_{i=0}^n$ such that Eq. \eqref{wc} holds $\forall \rho \in \mathcal{O}(\mathbb{C}^d)$, where $\mathcal{O}(\mathbb{C}^d)$ denotes the set of quantum states in $\mathbb{C}^d$.
\end{definition}
%
If we have $\{(G',\Vec{w}),\{\Pi_i\}_{i=1}^{n'},\rho\}$ satisfying Eq. \eqref{wc} that includes projectors that are not of rank one, we can obtain $\{(G,\Vec{w}),\{|\psi_i\rangle\!\langle \psi_i|\}_{i=1}^{n},\rho\}$ satisfying Eq. \eqref{wc} by splitting each of the projectors that are not rank one into rank-one projectors. See Appendix \ref{app:lemma}.


\begin{figure}[h!]
	\includegraphics[width=0.47\textwidth]{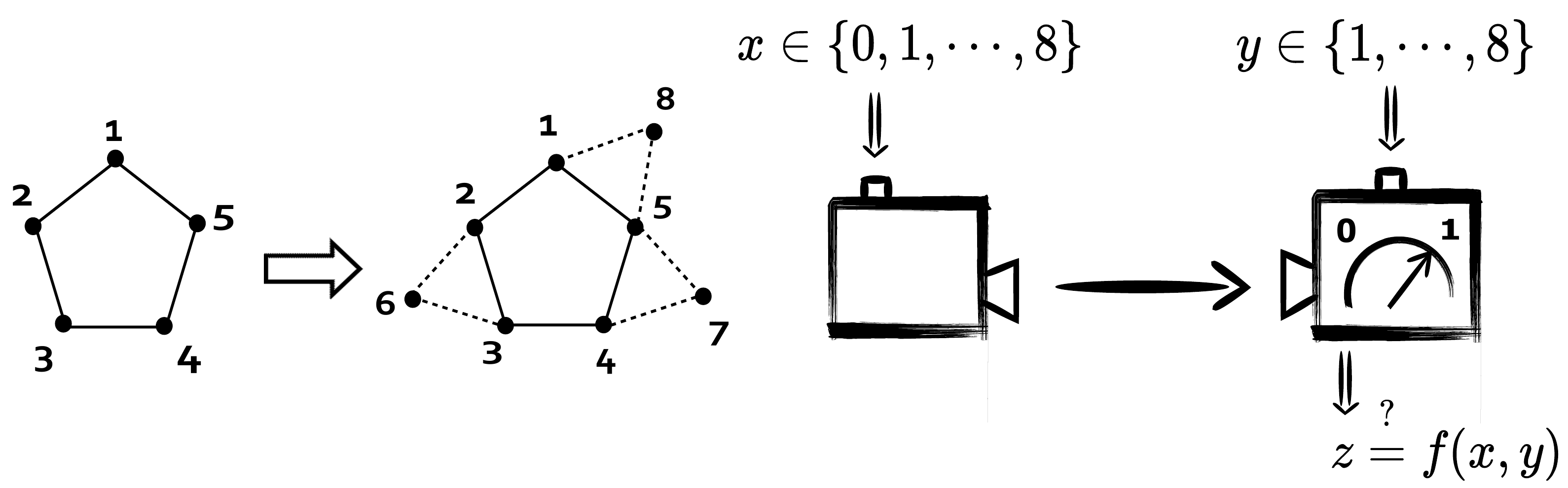}
	\caption{On the left, the construction of the extended graph from the 5-cycle graph. Each of the 8 vertices of the extended graph belongs to at least one clique of size $3$. On the right, scheme of the communication complexity task based on the extended graph.}
	\label{task}
\end{figure}


\textit{One-way communication complexity.---}Communication complexity \cite{ccbook} studies the amount of communication required for tasks involving inputs distributed among several parties. In one-way communication complexity \cite{ccbook,DEWOLF,RevModPhys.82.665,buhrman1998quantum}, there are two parties.
As shown in Fig.~\ref{task}, in each round, Alice, receives a random input $x \in X$. Depending upon $x$, Alice sends a message (classical or quantum) to Bob. In addition, Bob receives a random input $y \in Y$. Using $y$ and the message received from Alice, Bob outputs $z$ which is Bob's guess about a certain function $f(x,y)$. After many rounds, they produce the probability $p(z|x, y)$ of $z$, given inputs $x$ and $y$.
The figure of merit of the task is given by
\begin{equation}
S = \sum_{x,y} t(x,y) p(z=f(x,y)|x,y),
\label{fom}
\end{equation}
where $t(x,y)\geqslant 0$ and $\sum_{x,y} t(x,y) = 1$. 
We are interested in two aspects. Firstly, the maximum value of $S$ that can be achieved under the restriction that the dimension of the  (classical or quantum) system communicated from Alice to Bob is $d$. 
Secondly, the minimum dimensional (classical or quantum) system required to communicate in order to achieve a certain value of $S$. Sharing prior classical randomness between Alice and Bob is allowed.


\textit{Communication complexity advantage based on quantum contextuality witnesses.---} Consider $\{(G,\Vec{w}), \{|\psi_i\rangle\!\langle \psi_i|\}_{i=1}^n,\rho\}$ satisfying Eq. \eqref{wc} and such that $\rho \in \mathcal{O}(\mathbb{C}^d)$.
Since $w_i \geqslant 0$, without
loss of generality, we can take $\max_i w_i = 1$. The task is defined as follows. 
First, we consider an extended graph $\widetilde{G}$ by adding additional vertices to $G$ such that 
each vertex in $\widetilde{G}$ belongs to, at least, one clique of size $d$. A clique is a set of vertices in which every pair are adjacent. We thereupon assign additional vectors (or rank-one projectors) to those additional vertices, so each vector belongs to at least one basis within the new set of vectors; see Fig.~\ref{task}.
Alice receives $x \in \{0,1, \ldots, n+k\}$ and Bob receives 
$y \in \{1, \ldots, n+k\}$, where $k$ number of vertices is added. Bob outputs his guess for 
\be \label{eq:fxy}
 f(x,y) = \begin{cases}
0, \ \text{ if } \ y=x, \\
1, \ \text{ if } \ y\in N_x, \\ 
0, \ \text{ if } \ y \in \{1,\dots,n\} \text{ and } x=0, 
\end{cases}
\ee 
where $N_x$ is the set of the vertices that are adjacent to (i.e., neighbors of) $x$ in $\widetilde{G}$. In other words, Bob needs to distinguish the runs where
$y=x$ and $x=0$ from the runs where $y\in N_x$. Except for these, whenever $y \neq$ $x$ and $y\notin N_x$, or, $x=0$ and $y\in \{n+1,\ldots,n+k\}$, the runs do not contribute to the figure of merit of the communication task. 
That is, the task for Alice and Bob is to maximize
\begin{widetext}
\begin{equation}\label{fomk}
S^{(\widetilde{G},\vec{w},d)} = \frac{1}{N} \left[\sum_{x=1}^{n+k}p(z=0|x,y=x) + \sum_{x=1}^{n+k} \sum_{y\in N_x} p(z=1|x,y) + \sum_{y=1}^n w_y p(z=0|x=0,y)\right],
\end{equation}
\end{widetext}
where
\be \label{eqN}
N = n+k+\sum_{x=1}^{n+k} |N_x| + \sum_{i=1}^n w_i,
\ee 
thus $\sum_{x,y}t(x,y)=1$. Alice and Bob must accomplish this task with the restriction that the dimension of the (classical or quantum) system communicated between them is $d$. Therefore, the communication task is fully specified by the value of $d$, the extended graph $\widetilde{G}$, and the weights $\vec{w}$. The important point is that there is quantum advantage in this communication task whenever  $d$ is ``sufficiently small'' in the sense that $d \leqslant \chi(G)$, where $\chi(G)$ is the chromatic number of the graph $G$  \footnote{Chromatic number is the smallest number of colors needed to color the vertices of $G$ so that no two adjacent vertices share the same color \cite{bondybook}}.

\begin{result}
For any $\{(G,\Vec{w}), \{|\psi_i\rangle\!\langle \psi_i|\}_{i=1}^n,\rho \}$ with $\rho\in \mathcal{O}(\mathbbm{C}^d)$ such that Eq. \eqref{wc} holds and $d\leqslant \chi(G)$, there exists a quantum strategy
for the communication complexity task defined by Eq. \eqref{fomk} that provides an advantage over any strategy in which the system communicated between Alice and Bob is classical.
\end{result}
In a general quantum strategy, let $\rho_x \in \mathcal{O}(\mathbbm{C}^d)$ denote the quantum state sent by Alice to Bob upon receiving $x$, and let $\{M_{0|y},M_{1|y} = \openone - M_{0|y}\}_{y=1}^{n+k}$ denote the quantum measurement Bob performs on $\rho_x$ to obtain $z$, upon receiving input $y$. Suppose that $\{(G,\Vec{w}), \{|\psi_i\rangle\!\langle \psi_i|\}_{i=1}^n,\rho \}$ is a contextuality witness satisfying Eq. \eqref{wc}. 
Since, the projectors in $\{|\psi_i\rangle\!\langle \psi_i|\}_{i=1}^n$ are of rank one, one can always add $k$ rank-one projectors $\{|\psi_i\rangle\!\langle \psi_i|\}_{i=n+1}^{n+k}$ so that the extended set $\{|\psi_i\rangle\!\langle \psi_i|\}_{i=1}^{n+k}$ has the relations of orthogonality given by $\widetilde{G}$. Given $\{(G,\Vec{w}), \{|\psi_i\rangle\!\langle \psi_i|\}_{i=1}^n,\rho \}$, Alice and Bob choose an extended set and apply the following strategy:
\begin{eqnarray} \label{qcp}
& \rho_0 = \rho,\ 
\rho_x = |\psi_x\rangle\!\langle \psi_x|, \quad x=1, \ldots, n+k, \nonumber \\
& M_{0|y} = |\psi_y\rangle\!\langle \psi_y|, \quad y=1, \ldots, n+k .
\end{eqnarray}
This way, $p(z=0|x,y=x)=p(z=1|x,y\in N_x)=1$, so the value of $S^{(\widetilde{G},\vec{w},d)}$ in Eq. \eqref{fomk} is
\begin{equation}
 \frac{1}{N} \left[n+k + \sum_{x=1}^{n+k} |N_x| + \sum_{i=1}^n w_i \tr ( \rho |\psi_i\rangle\!\langle \psi_i|) \right],
\label{fom1}
\end{equation} 
while communicating a (quantum) system of dimension $d$ between Alice and Bob.
In contrast to that,

%
\begin{thm}\label{thm:Sc}
Whenever $d\leqslant \chi(G)$, for any strategy in which the system communicated between Alice and Bob is a classical system of dimension $d$, the value of $S^{(\widetilde{G},\vec{w},d)}$ is upper bounded by
\be \label{fomclassical} 
S^{(\widetilde{G},\vec{w},d)} \leqslant S^{(\widetilde{G},\vec{w},d)}_c = \frac{1}{N}\left[n+k+\sum_{x=1}^{n+k}|N_x| +\alpha(G,\Vec{w}) - \delta \right],
\ee 
where $\delta$ is the minimum number of ``improperly colored'' vertices of $\widetilde{G}$ when $d$ colors are used to color all the vertices. A vertex is improperly colored if it has at least one neighbor sharing the same color.
\end{thm}
For a proof, see Appendix \ref{app:thms}. 
 Because of Eq.~\eqref{wc} and the fact that $\delta$ is non-negative, the expression in Eq. \eqref{fom1} is strictly larger than $S_c^{(\widetilde{G},\vec{w},d)}$. 

Let us suppose that $d_{\min}$ is the minimum dimension in which the set of projectors $\{|\psi_i\rangle\!\langle \psi_i|\}$ and $\rho$ can be realized such that $\{(G,\Vec{w}), \{|\psi_i\rangle\!\langle \psi_i|\}_{i=1}^n,\rho \}$ is a contextuality witness.
For any SI contextuality witness, $\chi(G)>d_{\min}$ \cite{Cabello;2011arXiv,Ramanathan2014,Cabello2015}. Therefore, whenever $d=d_{\min}$, there will be at least two adjacent vertices sharing the same color when $d$ colors are used to color the graph, implying $\delta \geqslant 2$. Moreover, in this case, $\rho_0$ can be any quantum state in $\mathcal{O}(\mathbbm{C}^d)$.

Explicit examples of the quantum advantage for communication complexity tasks based on some quantum SI contextuality sets are presented in Appendix \ref{app:qkd}, together with a proof of their robustness against white noise.


\textit{Certifying contextuality witness from communication complexity task.---}The quantum communication strategy given by Eq.~\eqref{qcp} is based on a contextuality witness. However, a general quantum strategy with advantage consists of a set of states $\{\rho_x\}_{x=0}^{n+k}$ acting on $\mathbbm{C}^d$ and a set of measurement $\{M_{0|y}\}_{y=1}^{n+k}$ so that the value of $S^{(\widetilde{G},\vec{w},d)}$ is greater than $S^{(\widetilde{G},\vec{w},d)}_c$. In general, such a strategy may not be related to contextuality witnesses. Nevertheless, the following theorem allows us to identify whether or not an unknown quantum communication strategy admits a contextuality witness.
\begin{thm}\label{thm:ccw}
For the above introduced communication task defined by $S^{(\widetilde{G},\vec{w},d)}$, the following condition holds:
\be \label{2terms}
\forall x,y,\ p(0|x,y=x)=p(1|x,y\in N_x)=1,
\ee 
if and only if $\{\rho_x\}$ is a set of rank-one projectors that has $\widetilde{G}$ as graph of orthogonality and $\rho_x = M_{0|x}$.
\end{thm}
For a proof, see Appendix~\ref{app:thms}. Therefore, Theorem~\ref{thm:ccw} presents operational criteria to certify a set of rank-one projectors satisfying orthogonality relations according to a graph. Note that the probabilities in Eq.~\eqref{2terms} are the first two terms of $S^{(\widetilde{G},\vec{w},d)}$. Consider the particular case of the task \eqref{fomk} in which $d = \chi(G)$ and an unknown quantum strategy comprising $\{\rho_x\}$,$\{M_{0|y}\}$ that provides greater value than $S_c^{(\widetilde{G},\vec{w},d)}$. First, it follows from Eq.~\eqref{fomclassical} that, in this case, $\sum_{y=1}^n w_y \tr ( \rho_0 M_{0|y}) > \alpha(G,\Vec{w})$ since $\delta =0$. In addition to that, if the first two terms in $S^{(\widetilde{G},\vec{w},d)}$ attain their algebraic values, then Theorem~\ref{thm:ccw} implies $\{(G,\vec{w}), \{M_{0|y}\}_{y=1}^n, \rho_0 \}$ must be a contextuality witness.


\textit{Increasing advantage in communication complexity.---}Here, we will consider only those contextuality witnesses where $\chi(G)>d_{\min}$. In these cases, it suffices to consider a simplified version of the above-described communication complexity task by taking the first two terms of $S^{(\widetilde{G},\vec{w},d)}$. Therefore, the figure of merit will be
\begin{equation}
S^G = \frac{1}{N} \left[\sum_{x=1}^{n}p(z=0|x,y=x) + \sum_{x=1}^{n} \sum_{y\in N_x} p(z=1|x,y) \right],
\label{simS}
\end{equation}
where 
$N = n+\sum_{x=1}^{n} |N_x|$.
Here, the communication problem is solely based on the exclusivity graph $G$ \cite{Saha2019}, and we do not need to consider additional inputs apart from the set of vertices of $G$. More importantly,
this is an equality problem as Bob guesses whether his input $y$ is equal to $x$ or not \cite{ccbook}. 

Let $Q(G)$ (or $C(G)$) be the minimum dimension of quantum system (or classical system) that should be communicated to achieve $S^G=1$. We are now interested in quantum advantages in terms of $C(G)$ and $Q(G) $. A quantum advantage in communication complexity implies $C(G)>Q(G) $, or, equivalently, $\log_2{[C(G)]}>\log_2{[Q(G)]}$ conventionally expressed in terms of classical and quantum bits. 
\begin{thm}\label{thm:cc}
Given any witness $\{(G,\Vec{w}), \{|\psi_i\rangle\!\langle \psi_i|\}_{i=1}^n,\rho \}$ where $|\psi_i\rangle \in \mathbbm{C}^{d_{\min}}$,
\be 
Q(G) \leqslant d_{\min}\ , \quad C(G) = \chi(G) .
\ee 
Moreover, $Q(G)$ is the minimum dimension $d$ such that there exists a set of projectors $\{\Pi_i\}$ acting on $\mathbbm{C}^d$ satisfying the orthogonality relations given by $G$. 
\end{thm}
A proof is provided in Appendix \ref{app:thms}. We can readily check that the quantum strategy, $\rho_x = M_{0|x}= |\psi_i\rangle\!\langle \psi_i|$, yields $S^G=1$. Thus, we have an advantage  whenever $\chi(G)>d_{\min}$.  In order to observe an increasing advantage, we need to consider products of graphs.
\begin{definition}[Inclusive graph product or co-normal product or disjunctive product or OR product $G \times H$]
    The vertex set of the \textit{inclusive graph product} of two graphs $G,H$ is $V(G) \times V(H)$. The edges of $G\times H$ are defined as $(i,j) \sim (i',j')$  iff $i\sim j$ or $i' \sim j'$. We denote by $G^m$ the $m$-times product of the same graph $G$ \cite{Feige1997,feige1995randomized}. 
\end{definition}
%
%
\begin{thm}\label{thm:qm/cm}
Given a graph $G$ with $n$ vertices, the ratio between classical and quantum communication complexities of $S^G$ based on $G^m$, that is, $C(G^m)/Q(G^m)$ increases polynomially with $m$,
%
\be \label{qm/cm2}
\quad \frac{C(G^m)}{Q(G^m)} \geqslant \left(\frac{\chi_f(G)}{d_{\min} } \right)^m,  \text{ for } m \in \mathbbm{N},
\ee
where $\chi_f(G)$ is the fractional chromatic number of $G$ \cite{Ramanathan2014}. For any graph, $\chi_f(G) \leqslant \chi(G)$.
\end{thm}
For a proof, see Appendix \ref{app:thms}. Since $\chi_f(G)/d_{\min} > 1$ for any quantum SI contextuality set in dimension $d_{\min}$ \cite{Ramanathan2014,Cabello2015}, the right-hand-side of \eqref{qm/cm2} can be arbitrarily large as $m$ increases \footnote{Note that $d_{\min}$ is lower bounded by the size of the maximum clique, i.e., the clique of largest size, of $G$.}. 
It follows from \eqref{qm/cm2} that the difference between the classical and quantum complexities for the equality task based on $G^m$ is lower bounded by 
$m \cdot \log_2 \left(\chi_f(G)/d_{\min}\right)$ bits,
which increases with $m$.
In Table~\ref{tab:my_label}, we present some explicit examples of the quantum advantage.


\begin{table}[h!]
    \centering
    \begin{tabular}{|c|c|c|c|c|}
    \hline 
      SI witness & $d_{\min}$ & $\chi_f(G)$ &  $C(G^m)/Q(G^m)$ from Eq.~\eqref{qm/cm2} \\ 
       with $n$ & & & so that $d_{\min}^m \sim 200$ qubits \\ 
         \hline 
         \hline 
         YO-13 \cite{Yu2012} & 3 &  35/11 & $\geqslant 6 \times 10^{13}$ \\
       Peres-33  \cite{Peres_1991} & 3 & 13/4 & $\geqslant 4 \times 10^{13}$ \\
       CEG-18 \cite{Cabello1996} & 4 & 9/2 & $\geqslant 3.4 \times 10^{7}$ \\
       Pauli-240 \cite{ZParXiv2022} & 8 & 15 & $\geqslant 1.9 \times 10^{18}$ \\
       Pauli-4320 \cite{ZParXiv2022} & 16 & 60 & $\geqslant 5 \times 10^{28}$ \\
       \hline 
    \end{tabular}
    \caption{In order to compare the quantum advantages originated from various SI contextuality witnesses, we have taken the value of $m$ for each set such that 200 qubits is sufficient to accomplish the respective equality problem. With respect to that, the lower bounds on the classical and quantum ratios have been obtained for various SI contextuality witnesses. 
}
    \label{tab:my_label}
\end{table}
Before proceeding to the next section, we point out an example of SI witness and the respective equality problem where the separation between the classical and quantum communication complexities grows exponentially with the dimension. 
Consider the set of vectors in $\mathbbm{C}^d$
of the form 
$
(1/\sqrt{d})\left[1,(-1)^{x_1}, \ldots, (-1)^{x_{d-1}} \right]^T,
$
where $x_i\in \{0,1\}$
such that in every vector the number of $x_i$ taking value $1$ is even. Note that there are $2^{d-2}$ such vectors in $\mathbbm{C}^d$, and let us denote this set by $\{|\phi_i\ra \}_{i=1}^{2^{d-2}}$. The graph, say $G_{N_d}$, representing the orthogonality relations for this set of vectors was introduced by Newman \cite{newman2004independent} and has been recently studied in the context of application of contextuality \cite{ZParXiv2022}. It turns out for any $d\geqslant 1128$ and divisible by 4,  $\{(G_{N_d},\vec{w}),\{|\phi_i\ra\!\la\phi_i|\}\}$ is SI contextuality witness where $w_i=1$ for all $i$ (see Appendix \ref{app:ls} 
for the proof). Remarkably, for the equality problem defined by \eqref{simS} with respect to $G_{N_d}$, we have
\be \label{ccrNG}
\frac{C(G_{N_d})}{Q(G_{N_d})} \geqslant \frac{1}{d}\left(\frac{2}{1.99}\right)^{d} .
\ee 
Thus, the gap between classical and quantum complexities is at least
$0.007 d- \log_2 d$ bits. The detailed proof of this fact is provided in Appendix \ref{app:ls}, which follows from the results by Frankl-R\"{o}dl \cite{FR}. 


\textit{Semi-device-independent quantum key distribution.---}Here, we propose that a QKD protocol based on quantum advantage in the communication complexity task introduced by $S^{(\widetilde{G},\vec{w},d)}$ in Eq.~\eqref{fomk} where $d$ is taken to be $d_{\min}$. Unlike fully device-dependent protocols \cite{BB84,Singh2017}, our protocol is semi-device-independent \cite{Marcin2011} involving two black boxes - Alice's preparation device and Bob's measurement device. We only assume that ($i$) the dimension of the degrees-of-freedom (of the physical system), in which the information is encoded, is bounded by $d_{\min}$, and ($ii$) the devices may share classical randomness but that is uncorrelated with the choices of the inputs $x,y$. 
The QKD protocol is as follows. After completing a large number of runs, Alice randomly chooses some runs and publicly announces her input $x$ so that Bob can verify that the obtained value of the figure of merit is greater than $S_c$. Thereby, Bob is ensured that the probabilities produced by his device cannot be simulated by classical systems under the aforementioned assumptions. Bob publicly announces his input $y$ for the remaining runs. Subsequently, Alice notes down $f(x,y)$ according to Eq. \eqref{eq:fxy} as the shared key. Whenever $y \notin \{x,N_x\}$, or, $y\in \{n+1,\dots,n+k\}$ and $x=0$, Alice publicly announces that the transmission is unsuccessful. \\
It is not difficult to show that such QKD protocol is secure against restricted eavesdroppers whenever the contextuality witness satisfies \textit{monogamy relations} that are proposed in \cite{Ramanathan2012}. The monogamy relation between two witnesses of contextuality realized on two separate degrees-of-freedom of any quantum state $\rho$ implies
\be \label{eq:mr}
\sum_{i = 1}^n w_i \tr \left(\rho (\Pi_i\otimes \I)\right) + \sum_{i = 1}^n w_i \tr \left(\rho (\I \otimes \overline{\Pi}_i) \right) \leqslant 2\alpha(G,\Vec{w}) ,
\ee 
for any $\Vec{w}$, where $\{\Pi_i\}$, $\{\overline{\Pi}_i\}$ realize the respective exclusivity graph $G$. Such relation holds for a large class of contextuality witnesses, including the well-known odd-cycle witnesses \cite{Ramanathan2012}. The QKD protocol is secure if the mutual information of Alice-Bob  is greater than the mutual information of Alice-Eve \cite{Marcin2011}, i.e., $I(A:B)>I(A:E)$, which for individual attacks and binary output implies 
$S_B > S_E $,
taking $S_B(S_E)$ be the value of $S^{(\widetilde{G},\vec{w},d)}$ obtained by Bob (Eve). Since Eve also knows the input $y$, she and Bob are in the same state to guess $f(x,y)$. Because of Theorem \ref{thm:ccw}, when Bob observes that the first two terms in $S^{(\widetilde{G},\vec{w},d)}$ attain their maximal values,
then $M_{0|y}$ are rank-one projectors realizing $\widetilde{G}$. We assume that $M_{0|y}$ for Eve also realizes $\widetilde{G}$. Now, even if Eve shares arbitrary quantum correlation with the preparation device of Alice, due to monogamy relation \eqref{eq:mr}, the following holds true:
\be 
\sum_{y=1}^n w_y p_B(0|x=0,y) + \sum_{y=1}^n w_y p_E(0|x=0,y) \leqslant 2 \alpha(G,\Vec{w}) .
\ee 
Taking the best possible scenario for Eve in which she also observes \eqref{2terms}, the above relation implies
\be 
S_B + S_E \leqslant 2 S^{(\widetilde{G},\vec{w},d)}_c .
\label{mono_relation}
\ee 
Therefore, whenever Alice-Bob obtains quantum advantage, that is, $S_B>S^{(\widetilde{G},\vec{w},d)}_c$, the protocol is secure against such eavesdropping. Subsequently, the key rate can be obtained by $r = I(A:B)-I(A:E)$ (see
Table~\eqref{table} in Appendix~\ref{app:qkd}).

In addition to the QKD protocol, these communication tasks can also be used to generate quantum randomness in the prepare-and-measure scenario \cite{randomness,randomness1,randomness2}. We have discussed this in Appendix \ref{app:rc}.


\textit{Conclusions.---}This work shows that all forms of quantum contextuality with sufficiently small dimension provide quantum advantage in distributed computation and in various communication protocols without requiring entanglement. In distributed computation, equality problems are essential for implementing large-scale circuits and data verification \cite{ccbook,TRbook,rao_yehudayoff_2020} (see Appendix \ref{app:ep}). We show the existence of a variant of the equality problem pertaining to every vertex-weighted graph with certain properties providing an advantage over classical communication. 

Considering equality problems defined by the graphs of a large class of contextuality witnesses including all quantum state-independent contextuality witnesses, we show that the communication complexity required to execute such problems in classical theory is larger than that in quantum theory. Moreover, 
the complexity advantage increases with an increase in the number of inputs, identifying a class of equality problems that can be solved only in quantum communication. 

As further applications of quantum contextuality driven communication tasks, we show how such
tasks can be used for semi-device-independent QKD, as
well as for the purpose of randomness generation. As interesting open problems for further work, we point out the possibility of extending the security proof of the QKD protocol to arbitrary individual eavesdropping strategies and finding optimal communication complexity advantages.
It would also be interesting to extend the link between quantum contextuality and quantum advantage in communication complexity tasks involving more than two parties, like, quantum fingerprinting \cite{quantumFP}. \\

\textit{Acknowledgements.---} D. S. acknowledges National Post-Doctoral Fellowship (PDF/2020/001682) for support. A. C. is supported by Project Qdisc (Project No.\ US-15097, Universidad de Sevilla), with FEDER funds, QuantERA grant SECRET, by MINECO (Project No.\ PCI2019-111885-2), and MICINN (Project No.\ PID2020-113738GB-I00).
A. S. M. acknowledges support from Project No. DST/ICPS/QuEST/2018/98 of the Department of Science \& Technology, Government of India.


\bibliography{ref} 

\begin{thebibliography}{55}%
\makeatletter
\providecommand \@ifxundefined [1]{%
 \@ifx{#1\undefined}
}%
\providecommand \@ifnum [1]{%
 \ifnum #1\expandafter \@firstoftwo
 \else \expandafter \@secondoftwo
 \fi
}%
\providecommand \@ifx [1]{%
 \ifx #1\expandafter \@firstoftwo
 \else \expandafter \@secondoftwo
 \fi
}%
\providecommand \natexlab [1]{#1}%
\providecommand \enquote  [1]{``#1''}%
\providecommand \bibnamefont  [1]{#1}%
\providecommand \bibfnamefont [1]{#1}%
\providecommand \citenamefont [1]{#1}%
\providecommand \href@noop [0]{\@secondoftwo}%
\providecommand \href [0]{\begingroup \@sanitize@url \@href}%
\providecommand \@href[1]{\@@startlink{#1}\@@href}%
\providecommand \@@href[1]{\endgroup#1\@@endlink}%
\providecommand \@sanitize@url [0]{\catcode `\\12\catcode `\$12\catcode
  `\&12\catcode `\#12\catcode `\^12\catcode `\_12\catcode `\%12\relax}%
\providecommand \@@startlink[1]{}%
\providecommand \@@endlink[0]{}%
\providecommand \url  [0]{\begingroup\@sanitize@url \@url }%
\providecommand \@url [1]{\endgroup\@href {#1}{\urlprefix }}%
\providecommand \urlprefix  [0]{URL }%
\providecommand \Eprint [0]{\href }%
\providecommand \doibase [0]{http://dx.doi.org/}%
\providecommand \selectlanguage [0]{\@gobble}%
\providecommand \bibinfo  [0]{\@secondoftwo}%
\providecommand \bibfield  [0]{\@secondoftwo}%
\providecommand \translation [1]{[#1]}%
\providecommand \BibitemOpen [0]{}%
\providecommand \bibitemStop [0]{}%
\providecommand \bibitemNoStop [0]{.\EOS\space}%
\providecommand \EOS [0]{\spacefactor3000\relax}%
\providecommand \BibitemShut  [1]{\csname bibitem#1\endcsname}%
\let\auto@bib@innerbib\@empty
\bibitem [{\citenamefont {Bell}(1966)}]{Bell:1966RMP}%
  \BibitemOpen
  \bibfield  {author} {\bibinfo {author} {\bibfnamefont {J.~S.}\ \bibnamefont
  {Bell}},\ }\href {\doibase 10.1103/RevModPhys.38.447} {\bibfield  {journal}
  {\bibinfo  {journal} {Rev. Mod. Phys.}\ }\textbf {\bibinfo {volume} {38}},\
  \bibinfo {pages} {447} (\bibinfo {year} {1966})}\BibitemShut {NoStop}%
\bibitem [{\citenamefont {Kochen}\ and\ \citenamefont
  {Specker}(1967)}]{Kochen:1967JMM}%
  \BibitemOpen
  \bibfield  {author} {\bibinfo {author} {\bibfnamefont {S.}~\bibnamefont
  {Kochen}}\ and\ \bibinfo {author} {\bibfnamefont {E.~P.}\ \bibnamefont
  {Specker}},\ }\href {\doibase 10.1512/iumj.1968.17.17004} {\bibfield
  {journal} {\bibinfo  {journal} {J. Math. Mech.}\ }\textbf {\bibinfo {volume}
  {17}},\ \bibinfo {pages} {59} (\bibinfo {year} {1967})}\BibitemShut {NoStop}%
\bibitem [{\citenamefont {Klyachko}\ \emph {et~al.}(2008)\citenamefont
  {Klyachko}, \citenamefont {Can}, \citenamefont {Binicio\u{g}lu},\ and\
  \citenamefont {Shumovsky}}]{Klyachko:2008PRL}%
  \BibitemOpen
  \bibfield  {author} {\bibinfo {author} {\bibfnamefont {A.~A.}\ \bibnamefont
  {Klyachko}}, \bibinfo {author} {\bibfnamefont {M.~A.}\ \bibnamefont {Can}},
  \bibinfo {author} {\bibfnamefont {S.}~\bibnamefont {Binicio\u{g}lu}}, \ and\
  \bibinfo {author} {\bibfnamefont {A.~S.}\ \bibnamefont {Shumovsky}},\ }\href
  {\doibase 10.1103/PhysRevLett.101.020403} {\bibfield  {journal} {\bibinfo
  {journal} {Phys. Rev. Lett.}\ }\textbf {\bibinfo {volume} {101}},\ \bibinfo
  {pages} {020403} (\bibinfo {year} {2008})}\BibitemShut {NoStop}%
\bibitem [{\citenamefont {Cabello}(2008)}]{Cabello2008}%
  \BibitemOpen
  \bibfield  {author} {\bibinfo {author} {\bibfnamefont {A.}~\bibnamefont
  {Cabello}},\ }\href {\doibase 10.1103/PhysRevLett.101.210401} {\bibfield
  {journal} {\bibinfo  {journal} {Phys. Rev. Lett.}\ }\textbf {\bibinfo
  {volume} {101}},\ \bibinfo {pages} {210401} (\bibinfo {year}
  {2008})}\BibitemShut {NoStop}%
\bibitem [{\citenamefont {Cabello}\ \emph {et~al.}(2014)\citenamefont
  {Cabello}, \citenamefont {Severini},\ and\ \citenamefont
  {Winter}}]{Cabello2014}%
  \BibitemOpen
  \bibfield  {author} {\bibinfo {author} {\bibfnamefont {A.}~\bibnamefont
  {Cabello}}, \bibinfo {author} {\bibfnamefont {S.}~\bibnamefont {Severini}}, \
  and\ \bibinfo {author} {\bibfnamefont {A.}~\bibnamefont {Winter}},\ }\href
  {\doibase 10.1103/PhysRevLett.112.040401} {\bibfield  {journal} {\bibinfo
  {journal} {Phys. Rev. Lett.}\ }\textbf {\bibinfo {volume} {112}},\ \bibinfo
  {pages} {040401} (\bibinfo {year} {2014})}\BibitemShut {NoStop}%
\bibitem [{\citenamefont {Budroni}\ \emph {et~al.}(2022)\citenamefont
  {Budroni}, \citenamefont {Cabello}, \citenamefont {Kleinmann}, \citenamefont
  {Larsson},\ and\ \citenamefont {G\"uhne}}]{Budroni:RMP2022}%
  \BibitemOpen
  \bibfield  {author} {\bibinfo {author} {\bibfnamefont {C.}~\bibnamefont
  {Budroni}}, \bibinfo {author} {\bibfnamefont {A.}~\bibnamefont {Cabello}},
  \bibinfo {author} {\bibfnamefont {M.}~\bibnamefont {Kleinmann}}, \bibinfo
  {author} {\bibfnamefont {J.-{\AA}.}\ \bibnamefont {Larsson}}, \ and\ \bibinfo
  {author} {\bibfnamefont {O.}~\bibnamefont {G\"uhne}},\ }\href
  {https://arxiv.org/abs/2102.13036} {\bibfield  {journal} {\bibinfo  {journal}
  {Rev. Mod. Phys.}\ } (\bibinfo {year} {2022})}\BibitemShut {NoStop}%
\bibitem [{\citenamefont {Ekert}(1991)}]{Ekert1991}%
  \BibitemOpen
  \bibfield  {author} {\bibinfo {author} {\bibfnamefont {A.~K.}\ \bibnamefont
  {Ekert}},\ }\href {\doibase 10.1103/PhysRevLett.67.661} {\bibfield  {journal}
  {\bibinfo  {journal} {Phys. Rev. Lett.}\ }\textbf {\bibinfo {volume} {67}},\
  \bibinfo {pages} {661} (\bibinfo {year} {1991})}\BibitemShut {NoStop}%
\bibitem [{\citenamefont {Buhrman}\ \emph {et~al.}(2010)\citenamefont
  {Buhrman}, \citenamefont {Cleve}, \citenamefont {Massar},\ and\ \citenamefont
  {de~Wolf}}]{RevModPhys.82.665}%
  \BibitemOpen
  \bibfield  {author} {\bibinfo {author} {\bibfnamefont {H.}~\bibnamefont
  {Buhrman}}, \bibinfo {author} {\bibfnamefont {R.}~\bibnamefont {Cleve}},
  \bibinfo {author} {\bibfnamefont {S.}~\bibnamefont {Massar}}, \ and\ \bibinfo
  {author} {\bibfnamefont {R.}~\bibnamefont {de~Wolf}},\ }\href {\doibase
  10.1103/RevModPhys.82.665} {\bibfield  {journal} {\bibinfo  {journal} {Rev.
  Mod. Phys.}\ }\textbf {\bibinfo {volume} {82}},\ \bibinfo {pages} {665}
  (\bibinfo {year} {2010})}\BibitemShut {NoStop}%
\bibitem [{\citenamefont {Cubitt}\ \emph {et~al.}(2010)\citenamefont {Cubitt},
  \citenamefont {Leung}, \citenamefont {Matthews},\ and\ \citenamefont
  {Winter}}]{Cubitt2010}%
  \BibitemOpen
  \bibfield  {author} {\bibinfo {author} {\bibfnamefont {T.~S.}\ \bibnamefont
  {Cubitt}}, \bibinfo {author} {\bibfnamefont {D.}~\bibnamefont {Leung}},
  \bibinfo {author} {\bibfnamefont {W.}~\bibnamefont {Matthews}}, \ and\
  \bibinfo {author} {\bibfnamefont {A.}~\bibnamefont {Winter}},\ }\href
  {\doibase 10.1103/PhysRevLett.104.230503} {\bibfield  {journal} {\bibinfo
  {journal} {Phys. Rev. Lett.}\ }\textbf {\bibinfo {volume} {104}},\ \bibinfo
  {pages} {230503} (\bibinfo {year} {2010})}\BibitemShut {NoStop}%
\bibitem [{\citenamefont {Kleinmann}\ \emph {et~al.}(2011)\citenamefont
  {Kleinmann}, \citenamefont {Gühne}, \citenamefont {Portillo}, \citenamefont
  {Larsson},\ and\ \citenamefont {Cabello}}]{Kleinmann2011}%
  \BibitemOpen
  \bibfield  {author} {\bibinfo {author} {\bibfnamefont {M.}~\bibnamefont
  {Kleinmann}}, \bibinfo {author} {\bibfnamefont {O.}~\bibnamefont {Gühne}},
  \bibinfo {author} {\bibfnamefont {J.~R.}\ \bibnamefont {Portillo}}, \bibinfo
  {author} {\bibfnamefont {J.-{\AA}.}\ \bibnamefont {Larsson}}, \ and\ \bibinfo
  {author} {\bibfnamefont {A.}~\bibnamefont {Cabello}},\ }\href {\doibase
  10.1088/1367-2630/13/11/113011} {\bibfield  {journal} {\bibinfo  {journal}
  {New J. Phys.}\ }\textbf {\bibinfo {volume} {13}},\ \bibinfo {pages} {113011}
  (\bibinfo {year} {2011})}\BibitemShut {NoStop}%
\bibitem [{\citenamefont {Howard}\ \emph {et~al.}(2014)\citenamefont {Howard},
  \citenamefont {Wallman}, \citenamefont {Veitch},\ and\ \citenamefont
  {Emerson}}]{Howard2014}%
  \BibitemOpen
  \bibfield  {author} {\bibinfo {author} {\bibfnamefont {M.}~\bibnamefont
  {Howard}}, \bibinfo {author} {\bibfnamefont {J.}~\bibnamefont {Wallman}},
  \bibinfo {author} {\bibfnamefont {V.}~\bibnamefont {Veitch}}, \ and\ \bibinfo
  {author} {\bibfnamefont {J.}~\bibnamefont {Emerson}},\ }\href {\doibase
  10.1038/nature13460} {\bibfield  {journal} {\bibinfo  {journal} {Nature}\
  }\textbf {\bibinfo {volume} {510}},\ \bibinfo {pages} {351} (\bibinfo {year}
  {2014})}\BibitemShut {NoStop}%
\bibitem [{\citenamefont {Cabello}\ \emph {et~al.}(2018)\citenamefont
  {Cabello}, \citenamefont {Gu}, \citenamefont {G\"uhne},\ and\ \citenamefont
  {Xu}}]{Cabello2018}%
  \BibitemOpen
  \bibfield  {author} {\bibinfo {author} {\bibfnamefont {A.}~\bibnamefont
  {Cabello}}, \bibinfo {author} {\bibfnamefont {M.}~\bibnamefont {Gu}},
  \bibinfo {author} {\bibfnamefont {O.}~\bibnamefont {G\"uhne}}, \ and\
  \bibinfo {author} {\bibfnamefont {Z.-P.}\ \bibnamefont {Xu}},\ }\href
  {\doibase 10.1103/PhysRevLett.120.130401} {\bibfield  {journal} {\bibinfo
  {journal} {Phys. Rev. Lett.}\ }\textbf {\bibinfo {volume} {120}},\ \bibinfo
  {pages} {130401} (\bibinfo {year} {2018})}\BibitemShut {NoStop}%
\bibitem [{\citenamefont {Grudka}\ \emph {et~al.}(2014)\citenamefont {Grudka},
  \citenamefont {Horodecki}, \citenamefont {Horodecki}, \citenamefont
  {Horodecki}, \citenamefont {Horodecki}, \citenamefont {Joshi}, \citenamefont
  {K\l{}obus},\ and\ \citenamefont {W\'ojcik}}]{Grudka2014}%
  \BibitemOpen
  \bibfield  {author} {\bibinfo {author} {\bibfnamefont {A.}~\bibnamefont
  {Grudka}}, \bibinfo {author} {\bibfnamefont {K.}~\bibnamefont {Horodecki}},
  \bibinfo {author} {\bibfnamefont {M.}~\bibnamefont {Horodecki}}, \bibinfo
  {author} {\bibfnamefont {P.}~\bibnamefont {Horodecki}}, \bibinfo {author}
  {\bibfnamefont {R.}~\bibnamefont {Horodecki}}, \bibinfo {author}
  {\bibfnamefont {P.}~\bibnamefont {Joshi}}, \bibinfo {author} {\bibfnamefont
  {W.}~\bibnamefont {K\l{}obus}}, \ and\ \bibinfo {author} {\bibfnamefont
  {A.}~\bibnamefont {W\'ojcik}},\ }\href {\doibase
  10.1103/PhysRevLett.112.120401} {\bibfield  {journal} {\bibinfo  {journal}
  {Phys. Rev. Lett.}\ }\textbf {\bibinfo {volume} {112}},\ \bibinfo {pages}
  {120401} (\bibinfo {year} {2014})}\BibitemShut {NoStop}%
\bibitem [{\citenamefont {Singh}\ \emph {et~al.}(2017)\citenamefont {Singh},
  \citenamefont {Bharti},\ and\ \citenamefont {Arvind}}]{Singh2017}%
  \BibitemOpen
  \bibfield  {author} {\bibinfo {author} {\bibfnamefont {J.}~\bibnamefont
  {Singh}}, \bibinfo {author} {\bibfnamefont {K.}~\bibnamefont {Bharti}}, \
  and\ \bibinfo {author} {\bibnamefont {Arvind}},\ }\href {\doibase
  10.1103/PhysRevA.95.062333} {\bibfield  {journal} {\bibinfo  {journal} {Phys.
  Rev. A}\ }\textbf {\bibinfo {volume} {95}},\ \bibinfo {pages} {062333}
  (\bibinfo {year} {2017})}\BibitemShut {NoStop}%
\bibitem [{\citenamefont {Saha}\ \emph {et~al.}(2019)\citenamefont {Saha},
  \citenamefont {Horodecki},\ and\ \citenamefont {Paw{\l}owski}}]{Saha2019}%
  \BibitemOpen
  \bibfield  {author} {\bibinfo {author} {\bibfnamefont {D.}~\bibnamefont
  {Saha}}, \bibinfo {author} {\bibfnamefont {P.}~\bibnamefont {Horodecki}}, \
  and\ \bibinfo {author} {\bibfnamefont {M.}~\bibnamefont {Paw{\l}owski}},\
  }\href {\doibase 10.1088/1367-2630/ab4149} {\bibfield  {journal} {\bibinfo
  {journal} {New J. Phys.}\ }\textbf {\bibinfo {volume} {21}},\ \bibinfo
  {pages} {093057} (\bibinfo {year} {2019})}\BibitemShut {NoStop}%
\bibitem [{\citenamefont {Bharti}\ \emph {et~al.}(2019)\citenamefont {Bharti},
  \citenamefont {Ray}, \citenamefont {Varvitsiotis}, \citenamefont {Warsi},
  \citenamefont {Cabello},\ and\ \citenamefont {Kwek}}]{bharti}%
  \BibitemOpen
  \bibfield  {author} {\bibinfo {author} {\bibfnamefont {K.}~\bibnamefont
  {Bharti}}, \bibinfo {author} {\bibfnamefont {M.}~\bibnamefont {Ray}},
  \bibinfo {author} {\bibfnamefont {A.}~\bibnamefont {Varvitsiotis}}, \bibinfo
  {author} {\bibfnamefont {N.}~\bibnamefont {Warsi}}, \bibinfo {author}
  {\bibfnamefont {A.}~\bibnamefont {Cabello}}, \ and\ \bibinfo {author}
  {\bibfnamefont {L.}~\bibnamefont {Kwek}},\ }\href {\doibase
  10.1103/PhysRevLett.122.250403} {\bibfield  {journal} {\bibinfo  {journal}
  {Phys. Rev. Lett.}\ }\textbf {\bibinfo {volume} {122}},\ \bibinfo {pages}
  {250403} (\bibinfo {year} {2019})}\BibitemShut {NoStop}%
\bibitem [{\citenamefont {Saha}\ \emph {et~al.}(2020)\citenamefont {Saha},
  \citenamefont {Santos},\ and\ \citenamefont {Augusiak}}]{Saha2020}%
  \BibitemOpen
  \bibfield  {author} {\bibinfo {author} {\bibfnamefont {D.}~\bibnamefont
  {Saha}}, \bibinfo {author} {\bibfnamefont {R.}~\bibnamefont {Santos}}, \ and\
  \bibinfo {author} {\bibfnamefont {R.}~\bibnamefont {Augusiak}},\ }\href
  {\doibase 10.22331/q-2020-08-03-302} {\bibfield  {journal} {\bibinfo
  {journal} {Quantum}\ }\textbf {\bibinfo {volume} {4}},\ \bibinfo {pages}
  {302} (\bibinfo {year} {2020})}\BibitemShut {NoStop}%
\bibitem [{\citenamefont {Emeriau}\ \emph {et~al.}(2022)\citenamefont
  {Emeriau}, \citenamefont {Howard},\ and\ \citenamefont
  {Mansfield}}]{PRXQuantum2022}%
  \BibitemOpen
  \bibfield  {author} {\bibinfo {author} {\bibfnamefont {P.-E.}\ \bibnamefont
  {Emeriau}}, \bibinfo {author} {\bibfnamefont {M.}~\bibnamefont {Howard}}, \
  and\ \bibinfo {author} {\bibfnamefont {S.}~\bibnamefont {Mansfield}},\ }\href
  {\doibase 10.1103/PRXQuantum.3.020307} {\bibfield  {journal} {\bibinfo
  {journal} {PRX Quantum}\ }\textbf {\bibinfo {volume} {3}},\ \bibinfo {pages}
  {020307} (\bibinfo {year} {2022})}\BibitemShut {NoStop}%
\bibitem [{Note1()}]{Note1}%
  \BibitemOpen
  \bibinfo {note} {There is other notion of contextuality, preparation
  contextuality \cite {Spekkens2005}, which has found application in oblivious
  multiplexing and state discrimination tasks \cite
  {Spekkens2009,Pan2018,PRXSchmid}}\BibitemShut {NoStop}%
\bibitem [{\citenamefont {Badzia\ifmmode~\mbox{\c{}}\else \c{}\fi{}g}\ \emph
  {et~al.}(2009)\citenamefont {Badzia\ifmmode~\mbox{\c{}}\else \c{}\fi{}g},
  \citenamefont {Bengtsson}, \citenamefont {Cabello},\ and\ \citenamefont
  {Pitowsky}}]{PhysRevLett.103.050401}%
  \BibitemOpen
  \bibfield  {author} {\bibinfo {author} {\bibfnamefont {P.}~\bibnamefont
  {Badzia\ifmmode~\mbox{\c{}}\else \c{}\fi{}g}}, \bibinfo {author}
  {\bibfnamefont {I.}~\bibnamefont {Bengtsson}}, \bibinfo {author}
  {\bibfnamefont {A.}~\bibnamefont {Cabello}}, \ and\ \bibinfo {author}
  {\bibfnamefont {I.}~\bibnamefont {Pitowsky}},\ }\href {\doibase
  10.1103/PhysRevLett.103.050401} {\bibfield  {journal} {\bibinfo  {journal}
  {Phys. Rev. Lett.}\ }\textbf {\bibinfo {volume} {103}},\ \bibinfo {pages}
  {050401} (\bibinfo {year} {2009})}\BibitemShut {NoStop}%
\bibitem [{\citenamefont {Yu}\ and\ \citenamefont {Oh}(2012)}]{Yu2012}%
  \BibitemOpen
  \bibfield  {author} {\bibinfo {author} {\bibfnamefont {S.}~\bibnamefont
  {Yu}}\ and\ \bibinfo {author} {\bibfnamefont {C.~H.}\ \bibnamefont {Oh}},\
  }\href {\doibase 10.1103/PhysRevLett.108.030402} {\bibfield  {journal}
  {\bibinfo  {journal} {Phys. Rev. Lett.}\ }\textbf {\bibinfo {volume} {108}},\
  \bibinfo {pages} {030402} (\bibinfo {year} {2012})}\BibitemShut {NoStop}%
\bibitem [{\citenamefont {Kleinmann}\ \emph {et~al.}(2012)\citenamefont
  {Kleinmann}, \citenamefont {Budroni}, \citenamefont {Larsson}, \citenamefont
  {G\"uhne},\ and\ \citenamefont {Cabello}}]{PhysRevLett.109.250402}%
  \BibitemOpen
  \bibfield  {author} {\bibinfo {author} {\bibfnamefont {M.}~\bibnamefont
  {Kleinmann}}, \bibinfo {author} {\bibfnamefont {C.}~\bibnamefont {Budroni}},
  \bibinfo {author} {\bibfnamefont {J.-{\AA}.}\ \bibnamefont {Larsson}},
  \bibinfo {author} {\bibfnamefont {O.}~\bibnamefont {G\"uhne}}, \ and\
  \bibinfo {author} {\bibfnamefont {A.}~\bibnamefont {Cabello}},\ }\href
  {\doibase 10.1103/PhysRevLett.109.250402} {\bibfield  {journal} {\bibinfo
  {journal} {Phys. Rev. Lett.}\ }\textbf {\bibinfo {volume} {109}},\ \bibinfo
  {pages} {250402} (\bibinfo {year} {2012})}\BibitemShut {NoStop}%
\bibitem [{\citenamefont {Kushilevitz}\ and\ \citenamefont
  {Nisan}(2006)}]{ccbook}%
  \BibitemOpen
  \bibfield  {author} {\bibinfo {author} {\bibfnamefont {E.}~\bibnamefont
  {Kushilevitz}}\ and\ \bibinfo {author} {\bibfnamefont {N.}~\bibnamefont
  {Nisan}},\ }\href@noop {} {\emph {\bibinfo {title} {Communication
  Complexity}}}\ (\bibinfo  {publisher} {Cambridge Univ Press},\ \bibinfo
  {address} {Cambridge, UK},\ \bibinfo {year} {2006})\BibitemShut {NoStop}%
\bibitem [{\citenamefont {Roughgarden}(2016)}]{TRbook}%
  \BibitemOpen
  \bibfield  {author} {\bibinfo {author} {\bibfnamefont {T.}~\bibnamefont
  {Roughgarden}},\ }\href {\doibase 10.1561/0400000076} {\bibfield  {journal}
  {\bibinfo  {journal} {\textit{Communication Complexity (for Algorithm
  Designers)}, Foundations and Trends® in Theoretical Computer Science}\
  }\textbf {\bibinfo {volume} {11}},\ \bibinfo {pages} {217} (\bibinfo {year}
  {2016})}\BibitemShut {NoStop}%
\bibitem [{\citenamefont {Rao}\ and\ \citenamefont
  {Yehudayoff}(2020)}]{rao_yehudayoff_2020}%
  \BibitemOpen
  \bibfield  {author} {\bibinfo {author} {\bibfnamefont {A.}~\bibnamefont
  {Rao}}\ and\ \bibinfo {author} {\bibfnamefont {A.}~\bibnamefont
  {Yehudayoff}},\ }\href {\doibase 10.1017/9781108671644} {\emph {\bibinfo
  {title} {Communication Complexity: and Applications}}}\ (\bibinfo
  {publisher} {Cambridge University Press},\ \bibinfo {year}
  {2020})\BibitemShut {NoStop}%
\bibitem [{\citenamefont {Paw\l{}owski}\ and\ \citenamefont
  {Brunner}(2011)}]{Marcin2011}%
  \BibitemOpen
  \bibfield  {author} {\bibinfo {author} {\bibfnamefont {M.}~\bibnamefont
  {Paw\l{}owski}}\ and\ \bibinfo {author} {\bibfnamefont {N.}~\bibnamefont
  {Brunner}},\ }\href {\doibase 10.1103/PhysRevA.84.010302} {\bibfield
  {journal} {\bibinfo  {journal} {Phys. Rev. A}\ }\textbf {\bibinfo {volume}
  {84}},\ \bibinfo {pages} {010302} (\bibinfo {year} {2011})}\BibitemShut
  {NoStop}%
\bibitem [{\citenamefont {Ramanathan}\ \emph {et~al.}(2012)\citenamefont
  {Ramanathan}, \citenamefont {Soeda}, \citenamefont
  {Kurzy\ifmmode~\acute{n}\else \'{n}\fi{}ski},\ and\ \citenamefont
  {Kaszlikowski}}]{Ramanathan2012}%
  \BibitemOpen
  \bibfield  {author} {\bibinfo {author} {\bibfnamefont {R.}~\bibnamefont
  {Ramanathan}}, \bibinfo {author} {\bibfnamefont {A.}~\bibnamefont {Soeda}},
  \bibinfo {author} {\bibfnamefont {P.}~\bibnamefont
  {Kurzy\ifmmode~\acute{n}\else \'{n}\fi{}ski}}, \ and\ \bibinfo {author}
  {\bibfnamefont {D.}~\bibnamefont {Kaszlikowski}},\ }\href {\doibase
  10.1103/PhysRevLett.109.050404} {\bibfield  {journal} {\bibinfo  {journal}
  {Phys. Rev. Lett.}\ }\textbf {\bibinfo {volume} {109}},\ \bibinfo {pages}
  {050404} (\bibinfo {year} {2012})}\BibitemShut {NoStop}%
\bibitem [{\citenamefont {Kurzy\ifmmode~\acute{n}\else \'{n}\fi{}ski}\ \emph
  {et~al.}(2014)\citenamefont {Kurzy\ifmmode~\acute{n}\else \'{n}\fi{}ski},
  \citenamefont {Cabello},\ and\ \citenamefont {Kaszlikowski}}]{Kurzy2014}%
  \BibitemOpen
  \bibfield  {author} {\bibinfo {author} {\bibfnamefont {P.}~\bibnamefont
  {Kurzy\ifmmode~\acute{n}\else \'{n}\fi{}ski}}, \bibinfo {author}
  {\bibfnamefont {A.}~\bibnamefont {Cabello}}, \ and\ \bibinfo {author}
  {\bibfnamefont {D.}~\bibnamefont {Kaszlikowski}},\ }\href {\doibase
  10.1103/PhysRevLett.112.100401} {\bibfield  {journal} {\bibinfo  {journal}
  {Phys. Rev. Lett.}\ }\textbf {\bibinfo {volume} {112}},\ \bibinfo {pages}
  {100401} (\bibinfo {year} {2014})}\BibitemShut {NoStop}%
\bibitem [{\citenamefont {Saha}\ and\ \citenamefont
  {Ramanathan}(2017)}]{Saha2017}%
  \BibitemOpen
  \bibfield  {author} {\bibinfo {author} {\bibfnamefont {D.}~\bibnamefont
  {Saha}}\ and\ \bibinfo {author} {\bibfnamefont {R.}~\bibnamefont
  {Ramanathan}},\ }\href {\doibase 10.1103/PhysRevA.95.030104} {\bibfield
  {journal} {\bibinfo  {journal} {Phys. Rev. A}\ }\textbf {\bibinfo {volume}
  {95}},\ \bibinfo {pages} {030104} (\bibinfo {year} {2017})}\BibitemShut
  {NoStop}%
\bibitem [{\citenamefont {Bondy}\ and\ \citenamefont
  {Murty}(1976)}]{bondybook}%
  \BibitemOpen
  \bibfield  {author} {\bibinfo {author} {\bibfnamefont {J.~A.}\ \bibnamefont
  {Bondy}}\ and\ \bibinfo {author} {\bibfnamefont {U.~S.~R.}\ \bibnamefont
  {Murty}},\ }\href@noop {} {\emph {\bibinfo {title} {Graph Theory with
  Applications}}}\ (\bibinfo  {publisher} {Macmillan},\ \bibinfo {address}
  {London},\ \bibinfo {year} {1976})\BibitemShut {NoStop}%
\bibitem [{\citenamefont {Cabello}(2016)}]{PhysRevA.93.032102}%
  \BibitemOpen
  \bibfield  {author} {\bibinfo {author} {\bibfnamefont {A.}~\bibnamefont
  {Cabello}},\ }\href {\doibase 10.1103/PhysRevA.93.032102} {\bibfield
  {journal} {\bibinfo  {journal} {Phys. Rev. A}\ }\textbf {\bibinfo {volume}
  {93}},\ \bibinfo {pages} {032102} (\bibinfo {year} {2016})}\BibitemShut
  {NoStop}%
\bibitem [{\citenamefont {Cabello}(2021)}]{PhysRevLett.127.070401}%
  \BibitemOpen
  \bibfield  {author} {\bibinfo {author} {\bibfnamefont {A.}~\bibnamefont
  {Cabello}},\ }\href {\doibase 10.1103/PhysRevLett.127.070401} {\bibfield
  {journal} {\bibinfo  {journal} {Phys. Rev. Lett.}\ }\textbf {\bibinfo
  {volume} {127}},\ \bibinfo {pages} {070401} (\bibinfo {year}
  {2021})}\BibitemShut {NoStop}%
\bibitem [{\citenamefont {{de Wolf}}(2002)}]{DEWOLF}%
  \BibitemOpen
  \bibfield  {author} {\bibinfo {author} {\bibfnamefont {R.}~\bibnamefont {{de
  Wolf}}},\ }\href {\doibase https://doi.org/10.1016/S0304-3975(02)00377-8}
  {\bibfield  {journal} {\bibinfo  {journal} {Theor. Comput. Sci.}\ }\textbf
  {\bibinfo {volume} {287}},\ \bibinfo {pages} {337} (\bibinfo {year}
  {2002})}\BibitemShut {NoStop}%
\bibitem [{\citenamefont {Buhrman}\ \emph {et~al.}(1998)\citenamefont
  {Buhrman}, \citenamefont {Cleve},\ and\ \citenamefont
  {Wigderson}}]{buhrman1998quantum}%
  \BibitemOpen
  \bibfield  {author} {\bibinfo {author} {\bibfnamefont {H.}~\bibnamefont
  {Buhrman}}, \bibinfo {author} {\bibfnamefont {R.}~\bibnamefont {Cleve}}, \
  and\ \bibinfo {author} {\bibfnamefont {A.}~\bibnamefont {Wigderson}},\ }in\
  \href@noop {} {\emph {\bibinfo {booktitle} {Proceedings of the Thirtieth
  Annual ACM Symposium on Theory of Computing}}}\ (\bibinfo {year} {1998})\
  p.~\bibinfo {pages} {63}\BibitemShut {NoStop}%
\bibitem [{Note2()}]{Note2}%
  \BibitemOpen
  \bibinfo {note} {Chromatic number is the smallest number of colors needed to
  color the vertices of $G$ so that no two adjacent vertices share the same
  color \cite {bondybook}}\BibitemShut {NoStop}%
\bibitem [{\citenamefont {Cabello}(2011)}]{Cabello;2011arXiv}%
  \BibitemOpen
  \bibfield  {author} {\bibinfo {author} {\bibfnamefont {A.}~\bibnamefont
  {Cabello}},\ }\href@noop {} {\enquote {\bibinfo {title} {State-independent
  quantum contextuality and maximum nonlocality},}\ } (\bibinfo {year}
  {2011}),\ \Eprint {http://arxiv.org/abs/1112.5149} {arXiv:1112.5149}
  \BibitemShut {NoStop}%
\bibitem [{\citenamefont {Ramanathan}\ and\ \citenamefont
  {Horodecki}(2014)}]{Ramanathan2014}%
  \BibitemOpen
  \bibfield  {author} {\bibinfo {author} {\bibfnamefont {R.}~\bibnamefont
  {Ramanathan}}\ and\ \bibinfo {author} {\bibfnamefont {P.}~\bibnamefont
  {Horodecki}},\ }\href {\doibase 10.1103/PhysRevLett.112.040404} {\bibfield
  {journal} {\bibinfo  {journal} {Phys. Rev. Lett.}\ }\textbf {\bibinfo
  {volume} {112}},\ \bibinfo {pages} {040404} (\bibinfo {year}
  {2014})}\BibitemShut {NoStop}%
\bibitem [{\citenamefont {Cabello}\ \emph {et~al.}(2015)\citenamefont
  {Cabello}, \citenamefont {Kleinmann},\ and\ \citenamefont
  {Budroni}}]{Cabello2015}%
  \BibitemOpen
  \bibfield  {author} {\bibinfo {author} {\bibfnamefont {A.}~\bibnamefont
  {Cabello}}, \bibinfo {author} {\bibfnamefont {M.}~\bibnamefont {Kleinmann}},
  \ and\ \bibinfo {author} {\bibfnamefont {C.}~\bibnamefont {Budroni}},\ }\href
  {\doibase 10.1103/PhysRevLett.114.250402} {\bibfield  {journal} {\bibinfo
  {journal} {Phys. Rev. Lett.}\ }\textbf {\bibinfo {volume} {114}},\ \bibinfo
  {pages} {250402} (\bibinfo {year} {2015})}\BibitemShut {NoStop}%
\bibitem [{\citenamefont {Feige}(1997)}]{Feige1997}%
  \BibitemOpen
  \bibfield  {author} {\bibinfo {author} {\bibfnamefont {U.}~\bibnamefont
  {Feige}},\ }\href {\doibase 10.1007/BF01196133} {\bibfield  {journal}
  {\bibinfo  {journal} {Combinatorica}\ }\textbf {\bibinfo {volume} {17}},\
  \bibinfo {pages} {79} (\bibinfo {year} {1997})}\BibitemShut {NoStop}%
\bibitem [{\citenamefont {Feige}(1995)}]{feige1995randomized}%
  \BibitemOpen
  \bibfield  {author} {\bibinfo {author} {\bibfnamefont {U.}~\bibnamefont
  {Feige}},\ }in\ \href@noop {} {\emph {\bibinfo {booktitle} {Proceedings of
  the Twenty-Seventh Annual ACM Symposium on Theory of Computing}}}\ (\bibinfo
  {publisher} {ACM},\ \bibinfo {address} {New York},\ \bibinfo {year} {1995})\
  p.\ \bibinfo {pages} {635}\BibitemShut {NoStop}%
\bibitem [{Note3()}]{Note3}%
  \BibitemOpen
  \bibinfo {note} {Note that $d_{\protect \qopname \relax m{min}}$ is lower
  bounded by the size of the maximum clique, i.e., the clique of largest size,
  of $G$.}\BibitemShut {Stop}%
\bibitem [{\citenamefont {Peres}(1991)}]{Peres_1991}%
  \BibitemOpen
  \bibfield  {author} {\bibinfo {author} {\bibfnamefont {A.}~\bibnamefont
  {Peres}},\ }\href {\doibase 10.1088/0305-4470/24/4/003} {\bibfield  {journal}
  {\bibinfo  {journal} {J. Phys. A: Math. Gen.}\ }\textbf {\bibinfo {volume}
  {24}},\ \bibinfo {pages} {L175} (\bibinfo {year} {1991})}\BibitemShut
  {NoStop}%
\bibitem [{\citenamefont {Cabello}\ \emph {et~al.}(1996)\citenamefont
  {Cabello}, \citenamefont {Estebaranz},\ and\ \citenamefont
  {García-Alcaine}}]{Cabello1996}%
  \BibitemOpen
  \bibfield  {author} {\bibinfo {author} {\bibfnamefont {A.}~\bibnamefont
  {Cabello}}, \bibinfo {author} {\bibfnamefont {J.}~\bibnamefont {Estebaranz}},
  \ and\ \bibinfo {author} {\bibfnamefont {G.}~\bibnamefont
  {García-Alcaine}},\ }\href {\doibase
  https://doi.org/10.1016/0375-9601(96)00134-X} {\bibfield  {journal} {\bibinfo
   {journal} {Phys. Lett. A}\ }\textbf {\bibinfo {volume} {212}},\ \bibinfo
  {pages} {183} (\bibinfo {year} {1996})}\BibitemShut {NoStop}%
\bibitem [{\citenamefont {Xu}\ \emph {et~al.}(2023)\citenamefont {Xu},
  \citenamefont {Steinberg}, \citenamefont {Singh}, \citenamefont
  {L{\'{o}}pez-Tarrida}, \citenamefont {Portillo},\ and\ \citenamefont
  {Cabello}}]{ZParXiv2022}%
  \BibitemOpen
  \bibfield  {author} {\bibinfo {author} {\bibfnamefont {Z.-P.}\ \bibnamefont
  {Xu}}, \bibinfo {author} {\bibfnamefont {J.}~\bibnamefont {Steinberg}},
  \bibinfo {author} {\bibfnamefont {J.}~\bibnamefont {Singh}}, \bibinfo
  {author} {\bibfnamefont {A.~J.}\ \bibnamefont {L{\'{o}}pez-Tarrida}},
  \bibinfo {author} {\bibfnamefont {J.~R.}\ \bibnamefont {Portillo}}, \ and\
  \bibinfo {author} {\bibfnamefont {A.}~\bibnamefont {Cabello}},\ }\href
  {\doibase 10.22331/q-2023-02-16-922} {\bibfield  {journal} {\bibinfo
  {journal} {{Quantum}}\ }\textbf {\bibinfo {volume} {7}},\ \bibinfo {pages}
  {922} (\bibinfo {year} {2023})}\BibitemShut {NoStop}%
\bibitem [{\citenamefont {Newman}(2004)}]{newman2004independent}%
  \BibitemOpen
  \bibfield  {author} {\bibinfo {author} {\bibfnamefont {M.~W.}\ \bibnamefont
  {Newman}},\ }\href@noop {} {\emph {\bibinfo {title} {Independent sets and
  eigenspaces}}}\ (\bibinfo  {publisher} {University of Waterloo},\ \bibinfo
  {year} {2004})\BibitemShut {NoStop}%
\bibitem [{\citenamefont {Frankl}\ and\ \citenamefont {Rodl}(1987)}]{FR}%
  \BibitemOpen
  \bibfield  {author} {\bibinfo {author} {\bibfnamefont {P.}~\bibnamefont
  {Frankl}}\ and\ \bibinfo {author} {\bibfnamefont {V.}~\bibnamefont {Rodl}},\
  }\href {http://www.jstor.org/stable/2000598} {\bibfield  {journal} {\bibinfo
  {journal} {Trans. Am. Math. Soc.}\ }\textbf {\bibinfo {volume} {300}},\
  \bibinfo {pages} {259} (\bibinfo {year} {1987})}\BibitemShut {NoStop}%
\bibitem [{\citenamefont {Bennett}\ and\ \citenamefont
  {Brassard}(2014)}]{BB84}%
  \BibitemOpen
  \bibfield  {author} {\bibinfo {author} {\bibfnamefont {C.~H.}\ \bibnamefont
  {Bennett}}\ and\ \bibinfo {author} {\bibfnamefont {G.}~\bibnamefont
  {Brassard}},\ }\href {\doibase https://doi.org/10.1016/j.tcs.2014.05.025}
  {\bibfield  {journal} {\bibinfo  {journal} {Theor. Comput. Sci.}\ }\textbf
  {\bibinfo {volume} {560}},\ \bibinfo {pages} {7} (\bibinfo {year}
  {2014})}\BibitemShut {NoStop}%
\bibitem [{\citenamefont {Li}\ \emph {et~al.}(2012)\citenamefont {Li},
  \citenamefont {Paw\l{}owski}, \citenamefont {Yin}, \citenamefont {Guo},\ and\
  \citenamefont {Han}}]{randomness}%
  \BibitemOpen
  \bibfield  {author} {\bibinfo {author} {\bibfnamefont {H.-W.}\ \bibnamefont
  {Li}}, \bibinfo {author} {\bibfnamefont {M.}~\bibnamefont {Paw\l{}owski}},
  \bibinfo {author} {\bibfnamefont {Z.-Q.}\ \bibnamefont {Yin}}, \bibinfo
  {author} {\bibfnamefont {G.-C.}\ \bibnamefont {Guo}}, \ and\ \bibinfo
  {author} {\bibfnamefont {Z.-F.}\ \bibnamefont {Han}},\ }\href {\doibase
  10.1103/PhysRevA.85.052308} {\bibfield  {journal} {\bibinfo  {journal} {Phys.
  Rev. A}\ }\textbf {\bibinfo {volume} {85}},\ \bibinfo {pages} {052308}
  (\bibinfo {year} {2012})}\BibitemShut {NoStop}%
\bibitem [{\citenamefont {Lunghi}\ \emph {et~al.}(2015)\citenamefont {Lunghi},
  \citenamefont {Brask}, \citenamefont {Lim}, \citenamefont {Lavigne},
  \citenamefont {Bowles}, \citenamefont {Martin}, \citenamefont {Zbinden},\
  and\ \citenamefont {Brunner}}]{randomness1}%
  \BibitemOpen
  \bibfield  {author} {\bibinfo {author} {\bibfnamefont {T.}~\bibnamefont
  {Lunghi}}, \bibinfo {author} {\bibfnamefont {J.~B.}\ \bibnamefont {Brask}},
  \bibinfo {author} {\bibfnamefont {C.~C.~W.}\ \bibnamefont {Lim}}, \bibinfo
  {author} {\bibfnamefont {Q.}~\bibnamefont {Lavigne}}, \bibinfo {author}
  {\bibfnamefont {J.}~\bibnamefont {Bowles}}, \bibinfo {author} {\bibfnamefont
  {A.}~\bibnamefont {Martin}}, \bibinfo {author} {\bibfnamefont
  {H.}~\bibnamefont {Zbinden}}, \ and\ \bibinfo {author} {\bibfnamefont
  {N.}~\bibnamefont {Brunner}},\ }\href {\doibase
  10.1103/PhysRevLett.114.150501} {\bibfield  {journal} {\bibinfo  {journal}
  {Phys. Rev. Lett.}\ }\textbf {\bibinfo {volume} {114}},\ \bibinfo {pages}
  {150501} (\bibinfo {year} {2015})}\BibitemShut {NoStop}%
\bibitem [{\citenamefont {Han}\ \emph {et~al.}(2016)\citenamefont {Han},
  \citenamefont {Yin}, \citenamefont {Li}, \citenamefont {Chen}, \citenamefont
  {Wang}, \citenamefont {Guo},\ and\ \citenamefont {Han}}]{randomness2}%
  \BibitemOpen
  \bibfield  {author} {\bibinfo {author} {\bibfnamefont {Y.-G.}\ \bibnamefont
  {Han}}, \bibinfo {author} {\bibfnamefont {Z.-Q.}\ \bibnamefont {Yin}},
  \bibinfo {author} {\bibfnamefont {H.-W.}\ \bibnamefont {Li}}, \bibinfo
  {author} {\bibfnamefont {W.}~\bibnamefont {Chen}}, \bibinfo {author}
  {\bibfnamefont {S.}~\bibnamefont {Wang}}, \bibinfo {author} {\bibfnamefont
  {G.-C.}\ \bibnamefont {Guo}}, \ and\ \bibinfo {author} {\bibfnamefont
  {Z.-F.}\ \bibnamefont {Han}},\ }\href {\doibase 10.1103/PhysRevA.93.032332}
  {\bibfield  {journal} {\bibinfo  {journal} {Phys. Rev. A}\ }\textbf {\bibinfo
  {volume} {93}},\ \bibinfo {pages} {032332} (\bibinfo {year}
  {2016})}\BibitemShut {NoStop}%
\bibitem [{\citenamefont {Buhrman}\ \emph {et~al.}(2001)\citenamefont
  {Buhrman}, \citenamefont {Cleve}, \citenamefont {Watrous},\ and\
  \citenamefont {de~Wolf}}]{quantumFP}%
  \BibitemOpen
  \bibfield  {author} {\bibinfo {author} {\bibfnamefont {H.}~\bibnamefont
  {Buhrman}}, \bibinfo {author} {\bibfnamefont {R.}~\bibnamefont {Cleve}},
  \bibinfo {author} {\bibfnamefont {J.}~\bibnamefont {Watrous}}, \ and\
  \bibinfo {author} {\bibfnamefont {R.}~\bibnamefont {de~Wolf}},\ }\href
  {\doibase 10.1103/PhysRevLett.87.167902} {\bibfield  {journal} {\bibinfo
  {journal} {Phys. Rev. Lett.}\ }\textbf {\bibinfo {volume} {87}},\ \bibinfo
  {pages} {167902} (\bibinfo {year} {2001})}\BibitemShut {NoStop}%
\bibitem [{\citenamefont {Spekkens}(2005)}]{Spekkens2005}%
  \BibitemOpen
  \bibfield  {author} {\bibinfo {author} {\bibfnamefont {R.~W.}\ \bibnamefont
  {Spekkens}},\ }\href {\doibase 10.1103/PhysRevA.71.052108} {\bibfield
  {journal} {\bibinfo  {journal} {Phys. Rev. A}\ }\textbf {\bibinfo {volume}
  {71}},\ \bibinfo {pages} {052108} (\bibinfo {year} {2005})}\BibitemShut
  {NoStop}%
\bibitem [{\citenamefont {Spekkens}\ \emph {et~al.}(2009)\citenamefont
  {Spekkens}, \citenamefont {Buzacott}, \citenamefont {Keehn}, \citenamefont
  {Toner},\ and\ \citenamefont {Pryde}}]{Spekkens2009}%
  \BibitemOpen
  \bibfield  {author} {\bibinfo {author} {\bibfnamefont {R.~W.}\ \bibnamefont
  {Spekkens}}, \bibinfo {author} {\bibfnamefont {D.~H.}\ \bibnamefont
  {Buzacott}}, \bibinfo {author} {\bibfnamefont {A.~J.}\ \bibnamefont {Keehn}},
  \bibinfo {author} {\bibfnamefont {B.}~\bibnamefont {Toner}}, \ and\ \bibinfo
  {author} {\bibfnamefont {G.~J.}\ \bibnamefont {Pryde}},\ }\href {\doibase
  10.1103/PhysRevLett.102.010401} {\bibfield  {journal} {\bibinfo  {journal}
  {Phys. Rev. Lett.}\ }\textbf {\bibinfo {volume} {102}},\ \bibinfo {pages}
  {010401} (\bibinfo {year} {2009})}\BibitemShut {NoStop}%
\bibitem [{\citenamefont {Ghorai}\ and\ \citenamefont {Pan}(2018)}]{Pan2018}%
  \BibitemOpen
  \bibfield  {author} {\bibinfo {author} {\bibfnamefont {S.}~\bibnamefont
  {Ghorai}}\ and\ \bibinfo {author} {\bibfnamefont {A.~K.}\ \bibnamefont
  {Pan}},\ }\href {\doibase 10.1103/PhysRevA.98.032110} {\bibfield  {journal}
  {\bibinfo  {journal} {Phys. Rev. A}\ }\textbf {\bibinfo {volume} {98}},\
  \bibinfo {pages} {032110} (\bibinfo {year} {2018})}\BibitemShut {NoStop}%
\bibitem [{\citenamefont {Schmid}\ and\ \citenamefont
  {Spekkens}(2018)}]{PRXSchmid}%
  \BibitemOpen
  \bibfield  {author} {\bibinfo {author} {\bibfnamefont {D.}~\bibnamefont
  {Schmid}}\ and\ \bibinfo {author} {\bibfnamefont {R.~W.}\ \bibnamefont
  {Spekkens}},\ }\href {\doibase 10.1103/PhysRevX.8.011015} {\bibfield
  {journal} {\bibinfo  {journal} {Phys. Rev. X}\ }\textbf {\bibinfo {volume}
  {8}},\ \bibinfo {pages} {011015} (\bibinfo {year} {2018})}\BibitemShut
  {NoStop}%
\end{thebibliography}%


\onecolumngrid

\appendix 


\section{Contextuality witness of rank-one projectors}\label{app:lemma}


\begin{lemma}
For any $\{(G',\Vec{w'}),\{\Pi_i\}_{i=1}^{n'},\rho\}$ satisfying Eq.~(2) in the main text,
where the rank of $\Pi_i$ is $r_i$ $\geqslant 1$, there is $\{(G,\Vec{w}),\{|\psi_{i,k}\rangle\!\langle \psi_{i,k}|\},\rho\}$ satisfying (2)
such that 
\be \label{splitPi}
\sum_{k=1}^{r_i} |\psi_{i,k}\rangle\!\langle \psi_{i,k}| = \Pi_i,
\ee
where $i = 1, \ldots, n'$ and, for each $i$, $k = 1, \ldots,$ $r_i$ and $(G',\Vec{w'})$ is a vertex-weighted subgraph of $(G,\Vec{w})$.
\end{lemma}
\begin{proof}
If the rank of $\Pi_i$ is $r_i (\geqslant 1)$, then we split the vertex $i$ into $r_i$ vertices, say $(i,k)$, where $k = 1, \ldots, r_i$. The edges and weights of the new graph $G$ are defined as follows:
\begin{eqnarray}
& \forall k, k', i, \quad  (i,k) \sim (i, k') \ , \text{ and } \nonumber \\
& \forall i, j, k, k', \quad  (i,k) \sim (j,k'), \ \ \text{\it iff}  \ \ i \sim j \ , \nonumber \\
& \forall i,k, \quad w_{(i,k)} = w'_i.  
\label{exclusiveR}
\end{eqnarray}
Eq.~(\ref{exclusiveR}) implies that $\alpha(G,\Vec{w}) = \alpha(G',\Vec{w'}) $. Now, if we split each projector according to \eqref{splitPi}, then
the set of rank-one projectors $ \{|\psi_{i,k}\rangle\!\langle \psi_{i,k}|\}$ has $G$ as the graph of orthogonality. Furthermore,
\begin{equation}
	\sum_{i=1}^n \sum_{k=1}^{r_i} w_{(i,k)}\ \tr (\rho |\psi_{i,k}\rangle\!\langle \psi_{i,k}|) = \sum_{i=1}^n w'_i\ \tr(\rho \Pi_i),
\end{equation} 
which is greater than $ \alpha(G,w)$ by definition.
\end{proof} 


\section{Communication complexity advantage based on quantum contextuality}\label{app:thms}


We first state a general feature of communication tasks in the following \textit{lemma}, which will be used later in the proof of the theorems.

\begin{lemma}
The maximum value of $S$ in Eq. (3),
when the system communicated between Alice and Bob is a classical system of dimension $d$, is given by
\be
S_c = \underset{\{p_e(m|x)\}}{\max} \left[1 - \sum_{y,m} \min \left\{ \sum_{x|f(x,y)=0} t(x,y) p_e(m|x), \sum_{x|f(x,y)=1} t(x,y) p_e(m|x) \right\} \right],
\label{reducedfom}
\ee 
where $m= 1,\dots,d$, and, $\forall x,m, \ p_e(m|x) \in \{0,1\}$, $\sum_m p_e(m|x) = 1$.
\end{lemma}
\begin{proof}
Alice and Bob share prior classical random variables $\lambda$ with arbitrary distribution $p(\lambda)$.
Alice's strategy is to encode the information of input $x$ into classical message $m$. A general description of such a strategy is given by a set of probabilities $\{p_e(m|x,\lambda)\}$, where $p_e(m|x,\lambda)$ is the probability of sending the message $m$ depending on input choice $x$ and random variable $\lambda$. Here $m \in \{1,\dots,d\}$, since the dimension of the communicated system is at most $d$. Bob's decoding strategy is generally represented by the set of probability distribution $\{p_d(z|y, m,\lambda)\}$, where $p_d(z|y,m,\lambda)$ is the probability of output $z$ given his input $y$, received message $m$, and the shared random variable $\lambda$. Therefore, the general expression of the probability of getting outcome $z$ upon receiving inputs $x,y$ in classical communication is given by
\be 
p(z|x,y) = \sum_m \sum_\lambda p(\lambda) p_e(m|x,\lambda) p_d(z|y,m,\lambda)  .
\ee 
By dividing all inputs $x$ into two subgroups according to the value of $f(x,y)$ given every $y$ and substituting each term using the above in the figure of merit, we find that
\begin{align} \label{calints}
S_c & = \max_{\substack{ \{p_e(m|x,\lambda)\}\\ \{p_d(z|y,m,\lambda)\}\\ \{p(\lambda)\}} }
\sum_y \bigg( \sum_{x|f(x,y)=0} t(x,y) p(0|x,y) + \sum_{x|f(x,y)=1} t(x,y) p(1|x,y) \bigg) \nonumber \\
& =   \max_{\substack{ \{p_e(m|x,\lambda)\} \\ \{p_d(z|y,m,\lambda)\}\\ \{p(\lambda)\}} } \sum_\lambda p(\lambda) \left[ \sum_{y,m} \left\{ \left(\sum_{x|f(x,y)=0} t(x,y) p_e(m|x,\lambda)\right) p_d(0|y,m,\lambda)
 +  \left(\sum_{ x|f(x,y)=1} t(x,y) p_e(m|x,\lambda)\right) p_d(1|y,m,\lambda) \right\} \right]. 
\end{align}
Let us see the term within the curly bracket $\{ \ldots \}$. Since $p_d(0|y,m,\lambda) + p_d(1|y,m,\lambda) =1$ the best decoding probability $p_d(z|y,m,\lambda)$ is fixed such that the above expression is 
\be 
S_c = \max_{\substack{ \{p_e(m|x,\lambda)\} \\ \{p(\lambda)\}} } \sum_\lambda p(\lambda) \left[ \sum_{y,m} \max \left\{ \sum_{x|f(x,y)=0} t(x,y) p_e(m|x,\lambda) ,  \sum_{ x|f(x,y)=1} t(x,y) p_e(m|x,\lambda) \right\} \right] .
\ee 
Given any encoding probability $\{p_e(m|x,\lambda)\}$, the above expression within the square bracket $[ \ldots ]$ is a convex function of $\lambda$, and thus, without loss of generality, we can take $p(\lambda)=1$ for which that expression is maximum or, equivalently, we can omit the dependence of $\lambda$. This implies
\begin{equation}\label{Sc}
S_c = \underset{\{p_e(m|x)\}}{\text{max}} \sum_{y, m} \max \left\{ \sum_{x|f(x,y)=0} t(x,y) p_e(m|x) ,  \sum_{ x|f(x,y)=1} t(x,y) p_e(m|x) \right\} .
\end{equation} 
Further, using the identity, $\max\{a,b\} = a + b - \min \{a,b\}$ for any non-negative number $a,b$, Eq. (\ref{Sc}) further reduces to
\be
S_c = \underset{\{p_e(m|x)\}}{\text{max}} \left[\sum_{y, m,x}   t(x,y) p_e(m|x) 
- \sum_{y,m}  \min \left\{ \sum_{x|f(x,y)=0} t(x,y) p_e(m|x) ,  \sum_{ x|f(x,y)=1} t(x,y) p_e(m|x) \right\} \right] .
\ee 
Using $\sum_{m} p_e(m|x) = 1, \forall x$, and $\sum_{x,y} t(x,y) =1$ in the above equation, we get \eqref{reducedfom}.
Finally, note that it is sufficient to consider the extremal values of $p_e(m|x)$ since the expression is convex. 
\end{proof}


\begin{proof}[Proof of Theorem 1]
Using Eq.~(\ref{reducedfom}) for $S^{(\widetilde{G},\vec{w},d)}$ given by Eq. (5),
we can express the optimal classical value as follows:
\bea
	S_c & = & \frac{1}{N} \underset{\{p_e(m|x)\}}{\text{max}} \Bigg[N - \sum_{y=1}^n \sum_{m=1}^d \text{min} \left\{p(m|x=y) + w_y p(m|x=0), \sum_{x \in N_y} p(m|x) \right\} \nonumber \\ 
	&& \qquad \qquad\qquad - \sum_{y=n+1}^{n+k} \sum_{m=1}^d \text{min} \left\{ p(m|x=y), \sum_{x\in N_y} p(m|x) \right\} \Bigg] ,
	\label{Sc1}
\eea 
wherein the encoding strategy $\{p_e(m|x)\}$ is deterministic, that is,
for each input $x$, Alice sends a classical level $m$ out of $d$ distinct levels. This is equivalent to assigning one of the $d$ colors to each input $x$.
Since $d \leqslant \chi(G)$, in general, $d$ colors are not sufficient to color all the vertices of the extended $\widetilde{G}$ properly. Say $\delta$ is the minimum number of `improperly colored' vertices when $d$ colors are used to color all the vertices of the extended exclusivity graph. A vertex is `improperly colored' if it has at least one neighbour that shares the same color.
We would like to find the optimal encoding strategy $\{p_e(m|x)\}$ that maximizes $S_c$ in \eqref{Sc1}. 
Say, the optimal encoding is such that, there are $\delta + q$ improperly colored vertices, that is, out of all the vertices, $n+k-\delta-q$ vertices are properly colored, where $q \in \{0, \ldots, n+k-\delta\}$ is an integer. Let us denote the level/color by $m_x$ that is
assigned to input $x=0,1,\dots,n+k$. There are three possibilities for each input $y$:
\begin{enumerate}
	\item \textit{y is not properly colored.} In this case there exists at least one $x \in N_y$ such that $p(m_y|x) = 1$. Therefore, 
	\be 
	\sum_{m} \text{min} \left\{p(m|x=y) + w_y p(m|x=0), \sum_{x \in N_y} p(m|x) \right\} \geqslant 1.
	\ee 
	
	\item \textit{y is properly coloured and $y \in \{1, \ldots, n\}$.} In this scenario, either $p(m_0|x=y)=0$, or, $p(m_0|x=y) =1$. For the former case, since every $y$ belongs to at least one $d$-clique there exists one $y \in N_x$, for which $p(m_0|x) =1$. Therefore,
	\be
	\text{min} \left\{p(m_0|x=y) + w_y p(m_0|x=0), \sum_{y\in N_x} p(m_0|x) \right\} \geqslant \min \{w_y,1\} = w_y .
	\ee 
	Here, we have used our convention for the contextuality witness that,
	\be \label{maxwi}
	\max_y w_y = 1. 
	\ee 
    While for the latter case,  
	\be 
	\text{min} \left\{p(m_0|x=y) + w_y p(m_0|x=0), \sum_{y\in N_x} p(m_0|x) \right\} = 0.
	\ee 
	
	\item  \textit{y is properly coloured and $y \in \{n+1, \ldots, n+k\}$.} In this case,  
	\be 
	\sum_m \text{min} \left\{p(m|x=y) , \sum_{y\in N_x} p(m|x) \right\} = 0.
	\ee 
\end{enumerate}
As the encoding strategy properly colors $n+k-\delta-q$ vertices, the above analysis implies
\be
 \sum_y \sum_{m} \text{min} \left\{p(m|x=y)
+w_y\ p(m|x=0), \sum_{x \in N_y} p(m|x) \right\}  \geqslant \delta + q + \sum\limits_{y|p(m_0|x=y) = 0} w_y.
\ee
Using this in Eq.~(\ref{Sc1}), we obtain 
\begin{equation}
	S_c \leqslant \underset{p_e(m|x)}{\text{max}} \frac{1}{N} \left[N - \delta -q-\sum_{y|p(m_0|y) = 0} w_y \right].
\end{equation}
Replacing $N$ by $n+k+\sum_{x=1}^{n+k}|N_x| + \sum_y w_y$ from Eq. (6)
in the above equation, we obtain
\begin{eqnarray}
	S_c &\leqslant &  \underset{p_e(m|x)}{\text{max}}\frac{1}{N} \left[n+k+\sum_{x=1}^{n+k} |N_x| -\delta - q + \sum_{y} w_y - \sum_{x|p(m_0|x=y)=0} w_y \right] \nonumber \\
	& = & \underset{p_e(m|x)}{\text{max}}\frac{1}{N} \left[n+k+\sum_{x=1}^{n+k} |N_x| - \delta - q + \sum_{x|p(m_0|y)=1} w_y \right].
\end{eqnarray}
If $q=0$, then 
\be 
\sum_{x|p(m_0|y)=1} w_y  \leqslant \alpha.
\ee 
On the other hand, if $q > 0$, then 
\be
\sum_{x|p(m_0|y)=1} w_y  \leqslant  \alpha + \sum_y w_y - q,
\ee 
where the sum over $y$ on right-hand-side is taken for $q$ number of different indices. It follows from Eq. \eqref{maxwi}
that the above is bounded by $\alpha + q \cdot (\max_y w_y) - q \leqslant \alpha .$
Therefore, the best strategy is choosing $q=0$, which implies Eq. (9).
\end{proof} 

\begin{proof}[Proof of Theorem 2]
  Eq. (10)
  implies 
  \be \label{myrx}
  \forall y, M_{0|y}^{\rho_{x=y}} = \I, \  M_{1|y}^{\rho_x} = \I, \text{ for } x\in N_y, 
  \ee 
  where $M_{z|y}^{\rho_x} $ denotes the reduced form of $M_{z|y}$ in the support of $\rho_x$. Note that $M_{0|x}^{\rho_{x}}$ are projectors and we must have 
  \be \label{rxrx'0}
  \forall x\sim x', \  (M_{0|x}^{\rho_{x}}) \ (M_{0|x'}^{\rho_{x'}}) = \mathbbm{O} , \text{ and } \rho_x \rho_{x'} = \mathbbm{O} .
  \ee 
 The size of the maximum clique in the extended graph $\widetilde{G}$ is $d$. Say, $\{1,\dots,d\}$ is a maximum clique. The conditions \eqref{rxrx'0} for any pair of $x,x' \in \{1,\dots,d\}$ and $\rho_x \in \mathbbm{C}^{d}$ hold true \textit{if and only if} $\rho_x$ are rank-one projectors and satisfy orthogonality relations according to $\widetilde{G}$. Moreover, it follows from \eqref{myrx} that 
 \be 
 M_{0|1}^{\rho_{1}} + \sum_{x=2}^{d} M_{1|1}^{\rho_{x}}  \geqslant \I_{d\times d} .
 \ee 
However, the normalization condition $M_{0|1}+M_{1|1}=\I_{d\times d}$ holds \textit{if and only if} $M^{\rho_x}_{0|x}=\rho_x$ for all $x=1,\dots,d$. Due to the fact that every input belongs to at least one maximum clique, the analysis holds true for all $x$. 
\end{proof} 


\begin{proof}[Proof of Theorem 3]
  $S^G$ given by Eq. (11)
  is $1$ if
  \be \label{sbar1}
  p(0|x,y=x) = p(1|x,y\in N_x) = 1 ,
  \ee 
  which implies Eq.~\eqref{myrx}-\eqref{rxrx'0}. Consequently, $\{M_{0|x}^{\rho_{x}}\}$ satisfies the orthogonality relations as per the exclusivity graph $G$. Moreover, the new set of measurements defined by $M_{0|x} = M_{0|x}^{\rho_x}$ also achieves $S^G=1$ for the same Alice's encoding strategy $\{\rho_x\}$. Therefore, for any quantum strategy that yields $S^G=1$, there exists a set of projectors that realizes $G$. In the reverse direction, for any set of projectors $\{\Pi_i\}$ realizing $G$, the strategy in which $M_{0|y}=\Pi_y, \rho_x = \Pi_x/\text{rank}(\Pi_x)$ achieves $S^G=1$. So, there is a one-to-one correspondence between quantum strategy achieving the perfect figure of merit and a set of projectors satisfying the orthogonality relations according to $G$. Thus, $Q(G)$ is the minimum dimension $d$ such that a set of projectors $\{\Pi_i\}$ acting on $\mathbbm{C}^d$ exists, and moreover, $Q(G)\leqslant d_{\min}$. \\
 In classical communication, it follows from \eqref{reducedfom} that 
  \be 
	S^G_c = \frac{1}{N} \underset{\{p_e(m|x)\}}{\text{max}} \left[N - \sum_{y=1}^n \sum_{m=1}^d \text{min} \left\{p(m|x=y), \sum_{x \in N_y} p(m|x) \right\} \right] ,
	\label{sSc1}
\ee 
which is 1 if and only if
\be 
\forall x,y, \ \text{min} \left\{p(m|x=y), \sum_{x \in N_y} p(m|x) \right\} = 0 .
\ee 
Since, without loss of generality, we can take $p_e(m|x)$ to be deterministic, the above implies if $p_e(m|x)=1$ for some $x$ then $\forall x'\in N_x, p_e(m|x')=0$. Therefore, the set of $\{m\}$ can be used to color the graph, and similarly, any set of colors used to color the graph can also be used as $\{m\}$ to achieve $S^G=1$. Thus, $C(G)$ is the minimum of number colors required to color the vertices of $G$.
\end{proof}

\begin{proof}[Proof of Theorem 4]
 To prove $Q(G^m) \leqslant (d_{\min})^m$, it suffices to show that there exists a set of projectors realizing $G^m$ in $d^m$ dimensional Hilbert space. Consider the realization where the vertex $(i_1,\dots,i_m)$ (where each $i_k\in\{1,\dots,n\}$) of $G^m$ is represented by the projector $\Pi_{i_1}\otimes \Pi_{i_2} \otimes \dots \otimes \Pi_{i_m}$. Clearly, if $i_k \sim j_{k}$ for any $k$, then the respective two projectors are orthogonal as $\Pi_{i_k}$ and $\Pi_{j_{k}}$ are orthogonal. \\
On the other hand, due to Theorem 3, we know that $C(G^m) = \chi(G^m)$. For any graph $G$, we have $\chi(G) \geqslant \chi_f(G)$. In addition, $\chi_f(G^m) = (\chi_f(G))^m$ according to Lemma 2.8 in Ref.~\cite{feige1995randomized}. Therefore, Eq. (13)
holds true. 
\end{proof} 

\textit{Remark.} We have another lower bound on $C(G^m)/Q(G^m)$ as follows.
\be \label{qm/cm} 
\frac{C(G^m)}{Q(G^m)} \geqslant \frac{2}{m\ln{n}} \left(\frac{\chi(G)}{d_{\min}}\right)^m,  \text{ for } m=2^q, q \in \mathbbm{N},
\ee
\begin{proof}
we make use of Theorem 2.3 of \cite{Feige1997} that states the following relation:
\be 
\chi(G\times H) \geqslant \frac{\chi(G)\cdot \chi(H)}{\ln{n}} 
\ee 
when both graphs $G$ and $H$ have same number of vertices $n$. Taking $H$ to be $G$, we have $\chi(G^2)\leqslant \chi(G)^2/\ln{n}$. Again, we taking product of two $G^2$ of the same number of vertices $n^2$, we have $\chi(G^4)\leqslant \chi(G)^4/\ln{n^2}$, and similarly, we can obtain 
\be 
\chi(G^m) \geqslant \frac{2\cdot \chi(G)^m}{m\ln{n}} 
\ee 
for any $m=2^q, q \in \mathbbm{N}$. Finally, using the facts that $C(G^m)= \chi(G^m), Q(G^m) \leqslant (d_{\min})^m$ and the above relation, we arrive at  \eqref{qm/cm}. 
\end{proof}


\section{SI contextuality witness and exponential separation in communication complexities}\label{app:ls}

Firstly, we point out the correspondence between the large classical vs. quantum communication complexities gap in distributed Deutsch-Jozsa task \cite{buhrman1998quantum,RevModPhys.82.665} and the equality problem (Eq. (11)) 
with respect to exclusivity graphs. Consider the set of vectors in $\mathbbm{C}^d$ of the form 
\be \label{xs}
(1/\sqrt{d})\left[(-1)^{x_1},(-1)^{x_2}, \ldots, (-1)^{x_d} \right]^T,
\ee 
where every $x_i\in \{0,1\}$. There is $2^d$ number of distinct vectors in this set, and two vectors from this set are orthogonal whenever values of $x_i$'s are different in exactly $d/2$ number of places. We can consider the graph, often called Hadamard graph $(G_{H_d})$, representing the orthogonality relations of the set of vectors in $\mathbbm{C}^d$, and thereupon, the equality problem defined by Eq.~(11) with respect to $G_{H_d}$.
The result by Frankl-R\"{o}dl \cite{FR} implies that $\alpha(G_{H_d}) \leqslant 1.99^d$, whenever $d$ is divisible by $4$ (Theorem~1.11 in \cite{FR}). Plugging this bound into the general relation $\chi(G)\geqslant n/\alpha(G)$, we find that for the Hadamard graph $\chi(G_{H_d})\geqslant (2/1.99)^d$. Therefore, the difference between classical and quantum communication complexities is at least
$0.007 d - \log_2 d$ bits. \\

We now consider the Newman graph \cite{newman2004independent} introduced in the main text. It is defined by the orthogonality relation of the set of vectors in $\mathbbm{C}^d$ that takes the form \eqref{xs}
where $x_1=0$ and in every vector the number of $x_i$ taking value $1$ is even. There are $2^{d-2}$ such vectors in $\mathbbm{C}^d$, and we denote this set by $\{|\phi_i\ra \}_{i=1}^{2^{d-2}}$.
Let us again take $d$ to be divisible by 4. It has been pointed out in \cite{ZParXiv2022} that
\be \label{thetaNG}
\sum_{i=1}^{2^d-2} |\phi_i\ra\!\la\phi_i| = \frac{2^{d-1}}{d} \I .
\ee 
On the other hand, due to \textit{Lemma 6.6.1} of \cite{newman2004independent} we know 
\be 
\alpha(G_{N_d}) = \frac{\alpha(G_{H_d})}{4}.
\ee 
Consequently, using the aforementioned result by Frankl-R\"{o}dl \cite{FR} we have
\be \label{alphaNG}
\alpha(G_{N_d}) \leqslant \frac{(1.99)^{d}}{4}.
\ee 
Comparing the above quantity with the right-hand-side of \eqref{thetaNG}, we see that $\{(G_{N_d},\vec{w}),\{|\phi_i\ra\!\la\phi_i|\}\}$ is SI contextuality witness (where $w_i=1$ for all $i$) if
\be 
\frac{2^{d-1}}{d} > \frac{(1.99)^{d}}{4} \implies d\geqslant 1128 . 
\ee 
It has been shown in \cite{ZParXiv2022} that $\{(G_{N_d},\vec{w}),\{|\phi_i\ra\!\la\phi_i|\}\}$ is SI contextuality witness for $d=28,32$.

In order to get a lower bound on the classical communication complexity of the respective equality problem, we plug the relation \eqref{alphaNG} into $\chi(G)\geqslant n/\alpha(G)$ and find that $\chi(G_{N_d})\geqslant (2/1.99)^{d}$. Finally, since $Q(G_{N_d})=d$ by its construction, we obtain Eq.~(14) in the main text.


\section{Examples}\label{app:qkd}


It follows from Eq.~(8) in the main text
that the maximum quantum value for $S^{(\widetilde{G},\vec{w},d)}$ from all $\{(G,\Vec{w}), \{|\psi_i\rangle\!\langle \psi_i|\}_{i=1}^n,\rho \}$, with $\rho \in \mathcal{O}(\mathbbm{C}^d) $, is given by
\be 
S^{(\widetilde{G},\vec{w},d)}_\beta := \frac{1}{N} \left[n+k + \sum_{x=1}^{n+k} |N_x| + \beta(G,\Vec{w}) \right],
\label{fomq}
\ee 
where 
\be \label{theta}
\beta(G,\Vec{w}) = \max\limits_{\{\{|\psi_i\rangle\!\langle \psi_i|\},\rho\}} \
\sum_{i=1}^n w_i \ \text{tr}(\rho |\psi_i\rangle\!\langle \psi_i|)
\ee 
and $\{|\psi_i\rangle\!\langle \psi_i|\}$ is a realization of $G$. 

\subsection{Robustness of the quantum advantage in communication tasks}

Here, we discuss the robustness of the quantum advantage in the presence of white noise. We take $\mu \in (0,1]$ to be the sharpness parameter or $(1-\mu)$ to be the parameter quantifying the amount of white noise. In such case, the communicated quantum state will be
\be
    \Tilde{\rho}_x = \mu \rho_x + (1-\mu) \frac{\I}{d} ,
\ee
where $\rho_x$ is the state that Alice wants to send given in Eq.~(7) in the main text. Taking into consideration Bob's measurements given by Eq.~(7) in the main text, the new probabilities are
\begin{eqnarray}
    p(0|x, y =x) &=& \mu + \frac{(1-\mu)}{d}, \nonumber \\ 
    p(1|x, y \in N_x) &=& \mu + \frac{(d-1)(1-\mu)}{d}, \nonumber \\
    p(0|x=0, y) &=& \mu \text{tr}(\rho |\psi_y\rangle \langle \psi_y|) + \frac{(1-\mu)}{d}.
\end{eqnarray}
Substituting these probabilities, we calculate the modified value of $S^{(\widetilde{G},\vec{w},d)}$ as equal to
\be
    \frac{1}{N}\Bigg[\mu(n+k) + \frac{(1-\mu)(n+k)}{d} + \mu \sum_{x=1}^{n+k}|N_x| + \frac{(1-\mu)(d-1)}{d}\sum_{x=1}^{n+k}|N_x| + \mu \sum_{y=1}^n w_y \text{tr}(\rho |\psi_y\rangle \langle \psi_y|) + \sum_{y=1}^n w_y \frac{(1-\mu)}{d}\Bigg] .
\ee
Now, we reckon the best quantum strategy defined in Eqs.~(\ref{fomq})--(\ref{theta}) to find the maximum modified quantum value in the presence of noise is
\be \label{qsrobust}
    \mu S^{(\widetilde{G},\vec{w},d)}_{\beta}+\frac{(1-\mu)}{dN} \Bigg[n+k+\sum_{x=1}^{n+k}|N_x|+\sum_{y=1}^n w_y + (d-2)\sum_{x=1}^{n+k}|N_x|\Bigg] .
\ee
Using the expression of $N$ from Eq.~(6) in the main text, we get a simplified form of the above expression,
\be \label{qsf_robust}
    \mu S^{(\widetilde{G},\vec{w},d)}_{\beta}+\frac{1-\mu}{d}+\frac{(1-\mu)(d-2)}{N}\sum_{x=1}^{n+k}|N_x|.
\ee
Subsequently, the quantum advantage persists this modified quantum value is greater than the best classical value, that is,
\be
    \mu S^{(\widetilde{G},\vec{w},d)}_{\beta}+\frac{1-\mu}{d}+\frac{(1-\mu)(d-2)}{N}\sum_{x=1}^{n+k}|N_x|>S^{(\widetilde{G},\vec{w},d)}_c ,
\ee
which, after a reshuffling of parameters, implies 
\be
    \mu>1-\frac{d N (S^{(\widetilde{G},\vec{w},d)}_{\beta}-S^{(\widetilde{G},\vec{w},d)}_c)}{dN S^{(\widetilde{G},\vec{w},d)}_{\beta}-N-(d-2)\sum_{x=1}^{n+k}|N_x|}:=\mu_c .
\ee
Here, $\mu_c$ is the critical value of $\mu$ (sharpness parameter) up to which there will not be any quantum advantage. In other words, to observe quantum advantage the sharpness parameter must be greater than $\mu_c$. The critical values of the sharpness parameter in our communication task for several well-known contextuality witnesses are computed in the third column of Table~\ref{table}.

\subsection{Communication complexity advantage based on extended 5-cycle graph}


Consider the first figure of the 5-cycle graph in Table~\ref{table}. There are total of five events and three extensions. So, $n = 5$ and $k = 3$. We take weighs $w_i = 1 \ \forall i\in \{1,\ldots,5\}$. The classical bound for this contextuality witness $\alpha(G, \Vec{w}) = 2$ whereas quantum bound $\beta(G, \Vec{w}) = \sqrt{5}$. 
Alice chooses an input $x$ from the set $[x] = \{0,1,\ldots,8\}$. Bob chooses an input $y$ from $[y] = \{1,2,\ldots,8\}$. The task is to maximize the figure of merit in Eq.~(5) of the main text.
For this, let us first determine $N$ using Eq.~(6) in the main text.
That is,
\begin{align}
N &= 5 + 3 + |N_1|+|N_2|+|N_3|+|N_4|+|N_5|+|N_6|+|N_7|+|N_8|+\sum_{i=1}^5 w_i \nonumber \\
 &= 5+3+3+3+3+3+4+2+2+2+5 = 35.
\end{align}
Using this and $\sum_{x=1}^8 |N_x| = 22$ in Eq.~(9) in the main text,
we get the maximum value of the figure of merit with classical settings as $S_c = 0.914$. 
On the other hand, the maximum figure of merit with quantum settings using Eq.~(\ref{fomq}) is $S_{\beta} = 0.921$. This shows a quantum advantage. Similarly, we have shown the quantum advantage in the considered communication task for other contextuality witnesses in the third column of Table~\ref{table}.


\subsection{Communication complexity advantages and key rates of the respective QKD protocol based on a few well-known quantum contextuality witnesses}


In order to calculate the average key rate per transmission, we compute the Shannon information of the transmitted string which is given as $E = -P_0 \log_2 (P_0) -P_1 \log_2 (P_1)$, where $P_0(P_1)$ is the probability of bit 0(1) in the key. Let $P_s$ be the success probability of generating the key, then the average key generation per transmission or the key rate is $P_s \cdot E$. The key rates for some contextuality witnesses are given in Table~\ref{table}. However, this key rate may not be secure.

The secure key rate ($\bf{r}$) is calculated using Eq. (17).
For this, we estimate the mutual information between Alice and Bob/Eve as follows \cite{Marcin2011}:
\be
I(A:X)= \sum_{j=0}^1 1 - h(S_{X_j}), \quad \text{where $X$ is $B$ or $E$.}
\ee
Here, $S_{B_j}(S_{E_j})$ is the figure of merit of guessing $f(x,y) \in \{0,1\}$ for Bob (Eve), $h(S_{X_j})$ is the Shannon binary entropy. We have taken $S_{B_j} = S_{\beta}$ and $S_{E_j} = S_c \forall \{0,1\}$. The secure key rates for some contextuality witnesses are also provided in  Table~\ref{table}. Note that Alice and Bob can apply the privacy amplification using Toeplitz matrix-based hash function on raw keys to make it secure.


\begin{table}
   \centering
\begin{tabular}{|c|c|c|c|}
\hline
NC & Graph & Communication complexity advantage & Quantum key distribution \\ 
\hline \hline
 5-cycle & \cell{\includegraphics[scale=0.3]{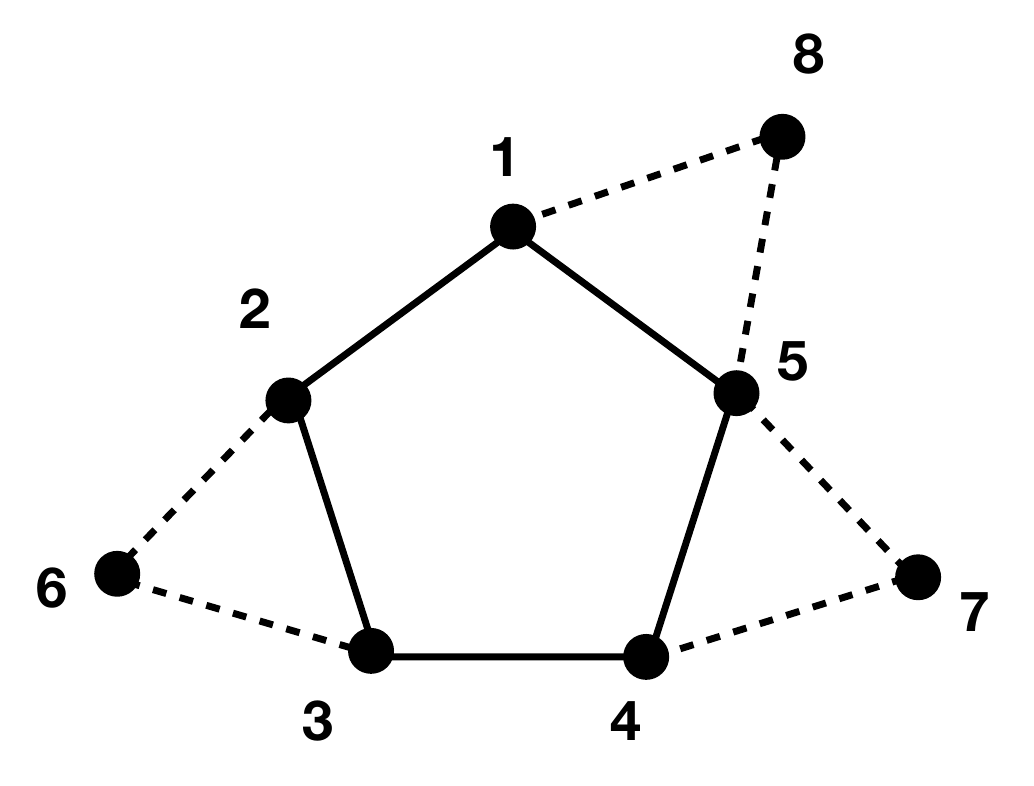}} & \cell{
 \small{$d=3,\ \chi(G)=3,\ \delta=0$} \\
 \small{$n=5, \ k=3 $} \\ 
 \small{$w_i=1, \ \forall i = 1,\dots,5$}
 \\
 \small{$\alpha(G,\Vec{w})=2, \quad \beta(G,\Vec{w}) = \sqrt{5}$} \\
 $N = 8 + ( 4 + 4\cdot 3 + 3\cdot 2 ) + 5 = 35$ \\
       \small{$S_c= \frac{1}{35}(8+22+2) = 0.914$ }  \\
       \small{$S_{\beta}= \frac{1}{35}(8+22+\sqrt{5}) = 0.921$}   \\
        \small{$\mu_c = 0.981$}}
  &  \cell{\small{$P_s
  = \frac{35}{72} 
  = 0.486$} \\
        \small{$P_0=\frac{(8+\sqrt{5})}{35} = 0.292$}   \\
        \small{$P_1=\frac{22}{35} = 0.628 $}\\
        \small{\textbf{Key rate = 0.457}}       \\  
        \small{\textbf{r = 0.049}}}       \\\hline
        
5-cycle* & \cell{\includegraphics[scale=0.3]{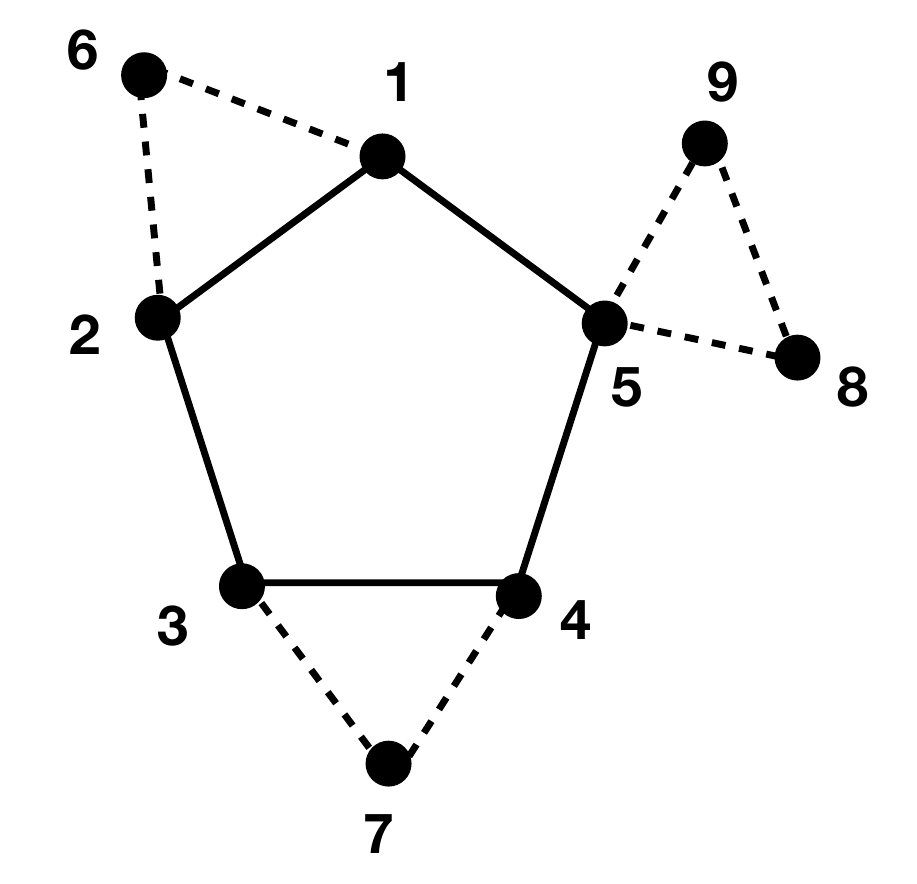}} & \cell{
 \small{$d=3,\ \chi(G)=3,\ \delta=0$} \\
\small{$n=5, \ k=4$} \\
\small{$w_i=1, \ \forall i = 1,\dots,5$}
 \\
$\alpha(G,\Vec{w})=2, \quad \beta(G,\Vec{w}) = \sqrt{5}$ \\
$N = 9 + ( 4 + 4\cdot 3 + 4\cdot 2 ) + 5 = 38$ \\
       $S_c= \frac{1}{38}(9+24+2) = 0.921$   \\
       $S_{\beta}= \frac{1}{38}(9+24+\sqrt{5}) = 0.927$ \\
       \small{$\mu_c = 0.984$}}
  &  \cell{\small{$P_s
  = \frac{38}{90} 
  = 0.422$} \\
        \small{$P_0=\frac{(9+\sqrt{5})}{38} = 0.296$}   \\
        \small{$P_1=\frac{24}{38} = 0.632$}\\
        \small{\textbf{Key rate = 0.396}}       \\  
        \small{\textbf{r = 0.043}}}       \\\hline
        
7-cycle & \cell{\includegraphics[scale=0.4]{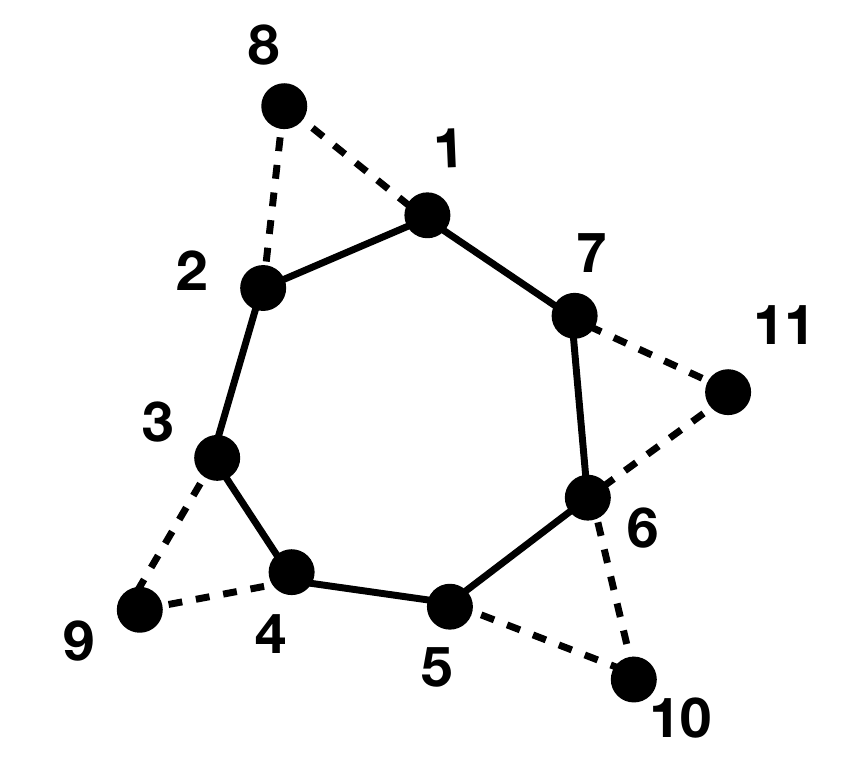}} & \cell{
 \small{$d=3,\ \chi(G)=3,\ \delta=0$} \\
\small{$n=7, \ k=4$ }\\
\small{$w_i=1, \ \forall i = 1,\dots,7$}
 \\
$\alpha(G,\Vec{w}) = 
3, \quad \beta(G,\Vec{w}) = \frac{7 \cos{\pi/7}}{1+\cos{\pi/7}}$ \\
$N = 11 + ( 4 + 6\cdot 3 + 4\cdot 2 ) + 7 = 48$ \\
        $S_c= \frac{1}{48}(11+30+3) = 0.917$   \\
    $S_{\beta}= \frac{1}{48}(11+30+\frac{7 \cos \pi/7}{1+\cos \pi/7}) = 0.923$\\
    \small{$\mu_c = 0.984$}}
    
  &  \cell{\small{$P_s
  = \frac{48}{132} = 0.364$} \\
        $P_0=\frac{\Big(11+\frac{7 \cos{\pi/7}}{1+\cos{\pi/7}}\Big)}{48} = 0.298$   \\
        \small{$P_1=\frac{30}{48} = 0.625$}\\
        \small{\textbf{Key rate = 0.343}}       \\  
        \small{\textbf{r = 0.042}}}       \\\hline
        
 CEG-18 & \cell{\includegraphics[scale=0.3]{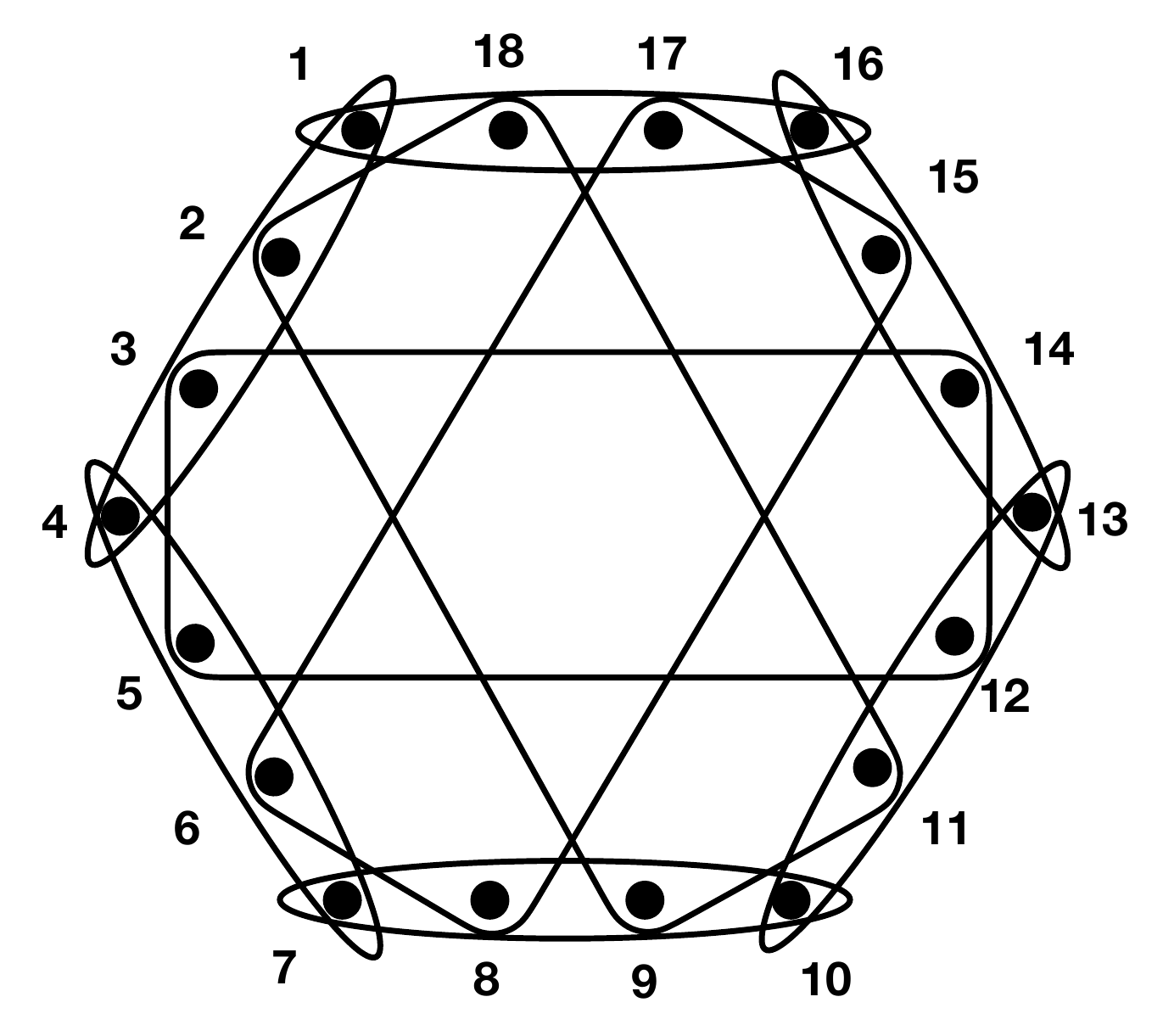}} & \cell{
  \small{$d=4,\ \chi(G)=5,\ \delta=3$} \\
 \small{$n=18, \ k=0$} \\
 \small{$w_i=1, \ \forall i = 1,\dots,18$}
 \\
 $\alpha(G,\Vec{w}) = 
 4, \quad \beta(G,\Vec{w}) = 9/2$ \\
 $N = 18 + ( 18\cdot 6 ) + 18 = 144$ \\
        \small{$S_c= \frac{1}{144}(18+108+4-3) = 0.882$}   \\
        $S_{\beta}= \frac{1}{144}(18+108+\frac{9}{2}) = 0.906$  \\
     \small{$\mu_c = 0.914$}}
  &  \cell{\small{$P_s
  =\frac{144}{342}
  = 0.421$}   \\
        \small{$P_0=\frac{(18+9/2)}{144}
         = 0.156$}   \\
        \small{$P_1=\frac{108}{144}
     = 0.75$}\\
        \small{\textbf{Key rate = 0.307}}}       \\  \hline
 
 YO-13 & \cell{\includegraphics[scale=0.35]{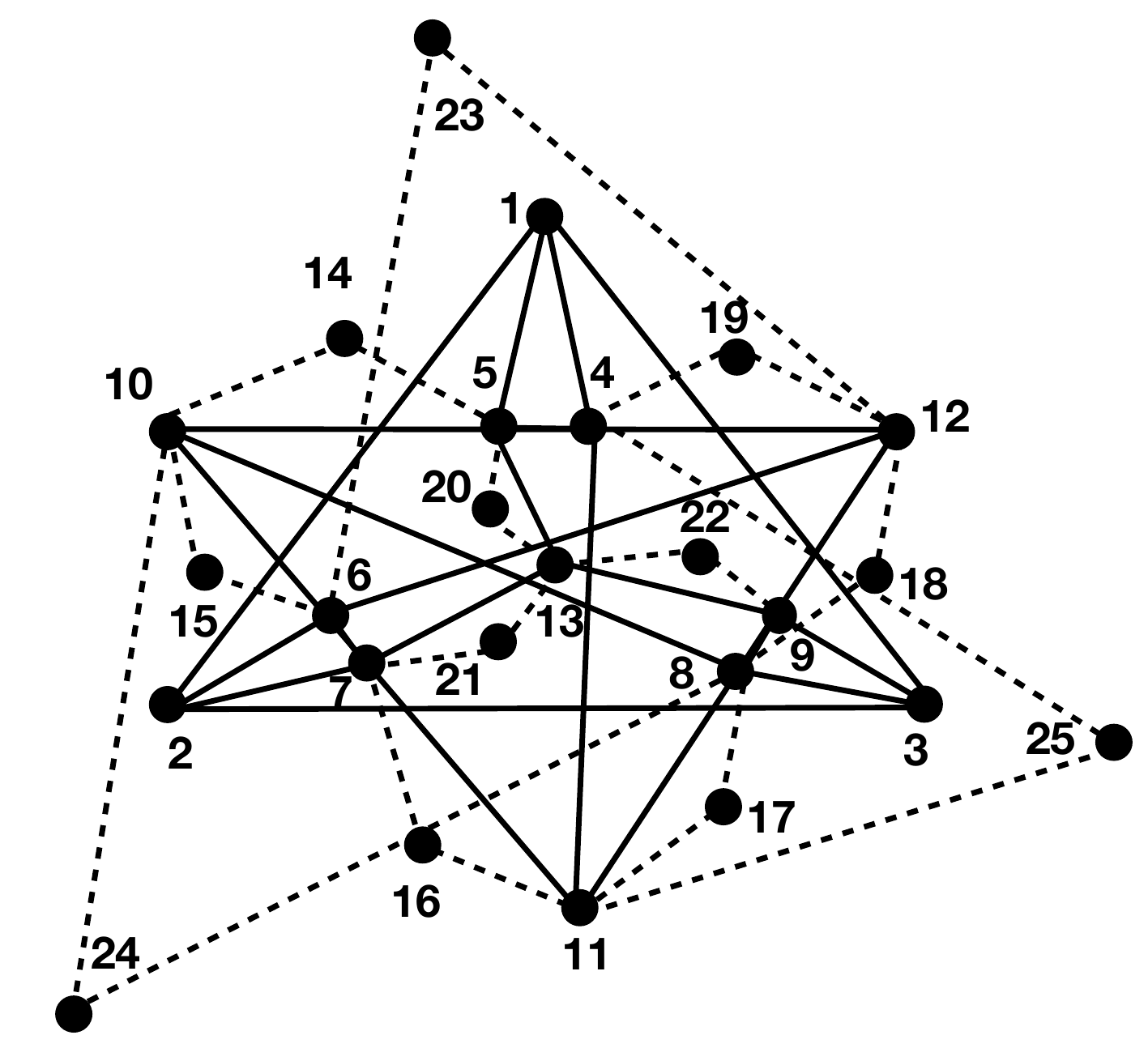}} & \cell{
  \small{$d=3,\ \chi(G)=4,\ \delta=2$} \\
 \small{$n=13, \ k=12$} \\
 \small{$w_i=1, \ \forall i = 1,\dots,9$}
 \\
 \small{$w_i=1/2, \ \forall i = 10,11,12,13$}
 \\
 $\alpha(G,\Vec{w}) = 
 7/2, \quad \beta(G,\Vec{w}) = 11/3 $ \\
 $N = 25 + ( 24 + 60 + 12 ) + 13 = 134$ \\
        \small{$S_c= \frac{1}{134}(25+96+\frac{7}{2}-2) = 0.914$}   \\
        \small{$S_{\beta}= \frac{1}{134}(25+96+\frac{11}{3}) = 0.930$} \\
        \small{$\mu_c = 0.955$}}
        
  &  \cell{\small{$P_s
  = \frac{134}{650} = 0.206$} \\
        \small{$P_0=\frac{(25+11/3)}{134} = 0.214$}   \\
        \small{$P_1=\frac{108}{134} = 0.806$}\\
        \small{\textbf{Key rate = 0.15}}}       \\\hline     
\end{tabular}
\caption{Quantum advantage in the one-way communication and quantum key distribution tasks described in the main text based on some quantum contextuality witnesses. The extended graphs are given in the second column, in which the additional edges are drawn by dashed lines. We have considered two different extensions of the 5-cycle graph in the first two rows. In the CEG-18 graph \cite{Cabello1996}, the nine 4-cliques are represented by the nine closed lines. In this case, no extension is needed. The details of the communication complexity task are given in the third column. The last column contains values of key rates and secure key rates ({\bf r}) in the presence of Eve for the corresponding semi-device independent QKD protocol. Note that we have not provided the values of {\bf r} for the CEG-18 and YO-13 \cite{Yu2012} graphs since these contextuality witnesses are not known to satisfy monogamy relation.}
\label{table}
\end{table} 

\newpage


\section{Randomness certification} \label{app:rc}

Schemes for quantum randomness generation have been proposed based on quantum advantages in communication tasks  \cite{randomness,randomness1,randomness2}. We can also use the communication complexity task introduced in Eq. (4)-(5)
to generate secure random bits from the untrusted preparation and measurement devices under the assumptions ($i$)-($ii$) mentioned in the preceding section. The random bits is obtained from the measurement outcome $z$.  Notice that for $x\neq 0$, all the probabilities appearing in the figure of merit (5)
are deterministic in the quantum strategies coming from contextuality witnesses. This enforces us to obtain the randomness only when $x=0$, which is suitably quantified by the minimum entropy function as follows:
\bea \label{Hinf}
&& \mathbf{H}_{\infty} = -\log_2 \big[ \max_{z,y} p(z|x=0,y) \big]  \\
    & \text{Subject to: } & \rho_x \in \mathbbm{C}^{d_{\min}}, \  S_o \in (S^{(\widetilde{G},\vec{w},d)}_{c},S^{(\widetilde{G},\vec{w},d)}_\beta ], \ \text{Eq. (10)}, 
    \nonumber
\eea
where $S_o$ is the obtained value of $S^{(\widetilde{G},\vec{w},d)}$. 
The above quantity is not necessarily nonzero for any contextuality witness, however, it is nonzero for odd-cycle contextuality witnesses \cite{bharti,Saha2020}.
In Figure II,
we evaluate $\mathbf{H}_{\infty}$ for the 5-cycle graph, which is non zero for all $S_o > S^{(\widetilde{G},\vec{w},d)}_c$ and attains the maximum value of $0.77$. 
\begin{figure}[h!]
	\includegraphics[width=0.45\textwidth]{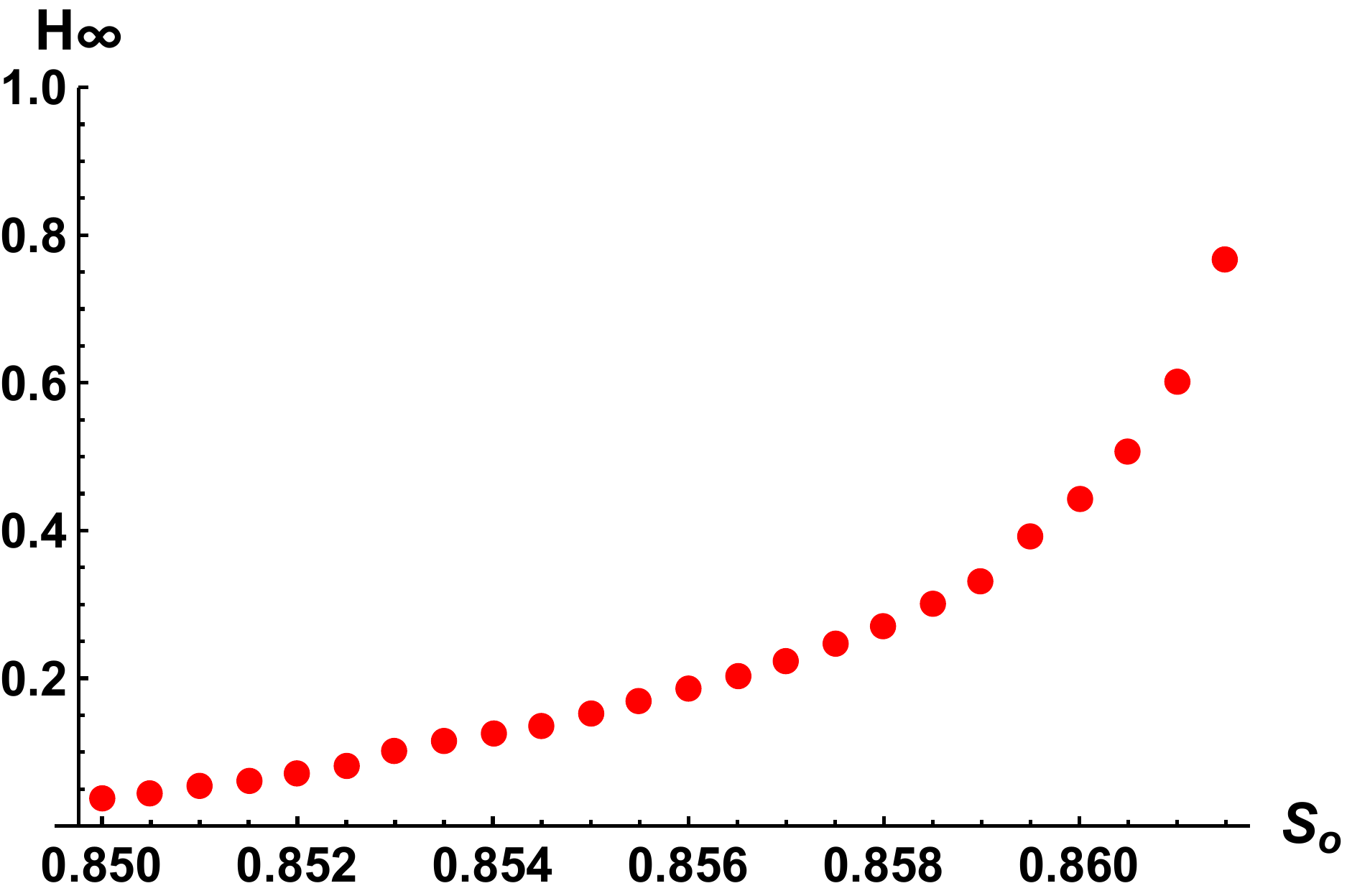}
\caption{
Vertical axis represents the randomness quantified by $\mathbf{H}_{\infty}$ and horizontal axis denotes the obtained figure of merit ($S_o \in (S_c, S_{\beta}]$). The contextual scenario under consideration is represented by the 5-cycle graph. 
The amount of randomness $H_{\infty}$ is found to be nonzero whenever the figure of merit is greater than the classical upper bound $S_c$. The maximum randomness ($H_{\infty} = 0.77$) corresponds to the optimal figure of merit in the quantum case $(S_{\beta})$. 
}
\label{randomfig}
\end{figure} 


\section{Practical applications of equality problems considered here}\label{app:ep}
The equality problems, given by Eq. (11), have direct practical applications in distributed computation. Same data stored in two stations get altered due to various reasons. Therefore, in distributed computation, it is often required to verify whether the data set in two sites is the same or not. We can consider $x$ and $y$ to be the data set variables in two stations, such that the variables can change to some specific variables due to the error. Now we can express this problem as an equality problem by a graph wherein the variables representing two adjacent vertices can interchange with each other. Communication complexity provides the minimum communication cost to verify whether the two variables are identical \cite{ccbook,rao_yehudayoff_2020}.\\
Apart from this, communication complexity of equality problems provides upper bounds on query complexity, the complexity of checking two variables are the same or not in a single device \cite{ccbook}. Streaming algorithms, in which all the inputs cannot be processed due to memory constraints, are also modeled as one-way communication complexity problems \cite{TRbook}. Furthermore, communication complexity of any equality problem has a natural application in game theory involving two agents \cite{TRbook}.


\end{document}